\def\llncs{0}
\def\fullpage{1}
\def\anonymous{0}
\def\draft{0}
\def\submission{0}
\def\disclosureonfirst{1}

\ifnum\submission=1
\def\llncs{1}
\def\draft{0}
\def\anonymous{1}
\def\fullpage{0}
\fi

\ifnum\llncs=1
    \documentclass{llncs}
    \ifnum\fullpage=1
    \usepackage{fullpage}
    \fi
\else
    \documentclass[letterpaper]{article}
    \ifnum\fullpage=1
    \usepackage{fullpage}
    \fi
    \usepackage{microtype}
\fi

\usepackage{algorithm}
\usepackage{algpseudocode}
\usepackage{verbatim}
\usepackage{authblk}
\usepackage{graphicx} 
\usepackage{physics} 
\usepackage{xcolor} 
\usepackage{amsmath}
\allowdisplaybreaks[4] 
\usepackage{amssymb}
\usepackage{mathtools}

 \usepackage{amsthm}
 \usepackage{thmtools}
\usepackage{enumitem} 
\usepackage{dsfont}
\usepackage[colorlinks=true,linkcolor=magenta,citecolor=blue,pagebackref=true,hypertexnames=false,pdftex,pdfpagelabels,bookmarks,hyperindex,hyperfigures]{hyperref}
\usepackage{dashbox}
\usepackage{fancybox,framed}
\usepackage[skins]{tcolorbox}
\usepackage{subcaption}
\usepackage{breakcites}
\usepackage{stmaryrd}
\usepackage{comment} 
\let\subparagraph\paragraph
\usepackage{titlesec}
\titleformat*{\paragraph}{\normalsize\bfseries}

\usepackage[capitalise,nameinlink]{cleveref}

\ifnum\draft=1
	\newcommand{\minki}[1]{\textcolor{blue}{$\langle\langle$Minki: #1$\rangle\rangle$}}
\else
	\newcommand{\minki}[1]{}
\fi

\ifnum\llncs=0
	\newtheorem{theorem}{Theorem}[section]
	\newtheorem{lemma}[theorem]{Lemma}
	\newtheorem{corollary}[theorem]{Corollary}
	
	\newtheorem{definition}[theorem]{Definition}

	\newtheorem{claim}[theorem]{Claim}
	\newtheorem{conjecture}[theorem]{Conjecture}

	\theoremstyle{remark}
	\newtheorem{remark}[theorem]{Remark}

\else
	\spnewtheorem{algorithmtheorem}{Algorithm}{\bfseries}{\rmfamily}

	\spnewtheorem{claim}{Claim}{\bfseries}{\itshape}
\fi
\crefname{appendix}{Appendix}{Appendices}
\Crefname{appendix}{Appendix}{Appendices}

\newcommand{\N}{\mathbb{N}}
\newcommand{\D}{\mathcal{D}}

\newcommand{\eps}{\varepsilon}

\renewcommand{\epsilon}{\varepsilon}

\newcommand{\oracle}{\mathcal{O}}

\ifnum\draft=1
	
\else
	
\fi

\title{Trapdoored Clifford Operators and Applications}
\date{}

\ifnum\llncs=1
    \ifnum\anonymous=1
        \author{}
        \institute{}
    \else
        \author{
        Hojune Lee\inst{1} \and
        Minki Hhan\inst{1}
        }
        \institute{
        KAIST, Daejeon, Korea
        \\\email{minkihhan@kaist.ac.kr}
        \\\email{hojunelee@kaist.ac.kr}
        }
    \fi
\else
    \author[1]{Minki Hhan}
    \author[1]{Hojune Lee}
    \affil[1]{{\small KAIST, Daejeon, Korea}
    \authorcr{\small minkihhan@kaist.ac.kr, hojunelee@kaist.ac.kr}}
    
\fi

\newcommand{\F}{\mathbb{F}}
\newcommand{\Mat}{\operatorname{Mat}}
\newcommand{\GL}{\operatorname{GL}}
\newcommand{\Upper}{\mathcal{U}}
\newcommand{\CS}{\mathcal{S}}

\newcommand{\Sym}{\operatorname{Sym}}

\newcommand{\negl}{\operatorname{negl}}
\newcommand{\CA}{\mathcal{A}}
\newcommand{\poly}{\operatorname{poly}}

\newcommand{\ql}{\widetilde{\oracle}(n)}
\newcommand{\unif}{\operatorname{Unif}}
\newcommand{\inv}{\mathsf{inv}}
\newcommand{\mult}{\mathsf{mult}}
\newcommand{\tp}{\mathsf{T}}

\newcommand{\CM}{\mathcal{M}}
\newcommand{\vect}[1]{\mathbf{{#1}}}

\newcommand{\Lower}{\mathcal{L}}
\newcommand{\Ber}{\mathsf{Ber}_{\F_q}}
\newcommand{\CH}{\mathcal{H}}
\newcommand{\CF}{\overline{\mathcal{F}}}
\newcommand{\nl}{\widetilde{\oracle}(n^{1 + \delta})}
\newcommand{\Pauli}{\overline{\mathcal{P}}_n}
\newcommand{\Clifford}{\overline{\mathcal{C}}_n}
\newcommand{\app}{\mathfrak{app}}
\newcommand{\tab}{\mathsf{tab}}
\newcommand{\strictupper}{\mathcal{U}_n^\circ}
\newcommand{\strictlower}{\mathcal{L}_n^\circ}
\newcommand{\CQ}{\mathcal{Q}}

\begin{document}

\maketitle

\begin{abstract}
    Random Clifford operators have numerous applications in quantum computing, including randomized benchmarking, classical shadows, and quantum authentication. However, sampling and implementing uniformly random $n$-qubit Clifford incur near-quadratic complexity due to the size of Clifford group.

We introduce a cryptographic way to overcome these barriers: trapdoored Clifford operator distributions whose samples are computationally indistinguishable from uniformly random Cliffords, yet implementing them can be much faster given the trapdoor. We construct a distribution of trapdoored Clifford operators whose elements can be sampled and implemented in near-linear time under a variant of the learning parity with noise assumption. 
Our constructions allow fast tableau action on Pauli labels for classical simulation, and also can be optimized to admit polylogarithmic-depth implementation.
Along the way, we construct trapdoored matrices over finite fields that support efficient multiplication by both a matrix and its inverse, resolving an open question left by Vaikuntanathan and Zamir [SODA'26]. 

We use these constructions to obtain faster protocols based on random Cliffords. We also explore their applications to the worst-case to average-case reductions for matrix and Clifford problems including the iterated matrix multiplication and Clifford circuit synthesis. 
In particular, we show the hardness of batching Clifford circuits: synthesizing circuits that apply the same Clifford to multiple registers is at least as hard as worst-case matrix multiplication, even when synthesis succeeds on a small constant fraction of random Cliffords. This extends to approximate implementations by general quantum circuits.
\end{abstract}
\ifnum\llncs=1
\fi

\ifnum\disclosureonfirst=1
\vfill
\paragraph{AI use disclosure.} 
AI use disclosure. ChatGPT (from 5.5 Pro to 6 Astra) was used during the preparation of this manuscript. Many of the constructions and proofs, as well as initial drafts of some sections, were generated by ChatGPT. The authors wrote the current manuscript based on these materials with AI assistance and take full responsibility for its content.
\fi


\ifnum\submission=1
\else
	\clearpage
	\newpage
	\setcounter{tocdepth}{2}
	\tableofcontents
	\newpage
\fi

\section{Introduction}

Clifford circuits, also known as stabilizer circuits and generated by Hadamard, phase, and controlled-NOT (CX) gates, play a central role in quantum computing.
Their algebraic structure enables efficient classical simulation via the stabilizer formalism~\cite{aaronson2004improved}.
On the other hand, adjoining a non-Clifford gate such as the $T$ gate yields a universal gate set for quantum computation.
These properties also make Clifford operators central to quantum error correction~\cite{Gottesman97} and fault-tolerant quantum computation~\cite{gottesman1998theory}.
Random Clifford operators are also widely used in various areas of quantum computing, including
randomized benchmarking~\cite{magesan2011scalable}, classical shadow estimation~\cite{huang2020predicting}, and quantum state tomography~\cite{kueng2016low}.
They also underpin constructions of random stabilizer codes~\cite{brown2013short} and quantum authentication schemes~\cite{aharonov2017interactive}.

Efficient sampling of uniformly random Clifford operators and their efficient circuit implementation have been extensively studied.
For $n$-qubit Cliffords, Koenig and Smolin~\cite{koenig2014efficiently} gave an $\oracle(n^3)$-time classical sampling algorithm.
Bravyi and Maslov~\cite{bravyi2021hadamard} achieved $\oracle(n^2)$ sampling time with circuits of size $\oracle(n^2)$, and van den Berg~\cite{van2021simple} subsequently gave a simpler method with the same complexity.
This is near-optimal 
as the optimal circuit size is known to be $\Theta(n^2/\log n)$ \cite[Corollary~9]{aaronson2004improved}.
This may still be unsatisfactory: even a single application of a uniformly random Clifford operator typically requires a nearly quadratic number of elementary gates and nearly linear depth.

In this paper, we ask whether this cost can be circumvented by relaxing uniform randomness to \emph{pseudorandomness}, i.e., indistinguishability from uniform randomness against efficient quantum algorithms.
This relaxation suffices for most applications, which involve only polynomial-time quantum computations.
We further require these operators to admit hidden, efficient circuit implementations, whose descriptions serve as \emph{trapdoors}.
We call such operators \emph{trapdoored Clifford operators}, a quantum analogue of the trapdoored matrices introduced in~\cite{braverman2025practical,vaikuntanathan2026improving}.

\subsection{Our Results}
Our main result is the construction of trapdoored Clifford operators under the cryptographic assumptions.

\begin{theorem}[Informal]
    Under the (dual) ring learning parity with noise (LPN) assumption, 
    we can sample (a description of) a trapdoored $n$-qubit Clifford operator $C$ together with its trapdoor $\mathsf{td}_C$ in $n^{1+\epsilon}$ time,
    where $\mathsf{td}_C$ gives an implementation of $C$ using $n^{1+\epsilon}$ elementary gates and $n^{1+\epsilon}$ clean ancillary qubits for arbitrarily small constant $\epsilon>0$.
    These can be reduced into $n \cdot \poly\log (n)$ assuming subexponential hardness.
\end{theorem}
Our trapdoored Clifford operators also support near-linear evaluation of the corresponding conjugate action on any Pauli label, providing a more efficient classical simulation. Furthermore, the Clifford circuit implementation can be done in a polylogarithmic depth under all-to-all connectivity.

A key ingredient of our trapdoored Clifford operators is the trapdoored matrices over finite field that supports efficient multiplication, inversion and more
under the same assumptions. 
This resolves an open question left in \cite[Section 3]{vaikuntanathan2026improving}.

\begin{theorem}[Informal]
    Under the (dual) ring learning parity with noise (LPN) assumption, 
    we can sample (a description of) a trapdoored invertible matrix $M\in\GL(n,\F_q)$ together with its trapdoor $\mathsf{td}_M$ in $n^{1+\epsilon}$ time for any $\epsilon>0$.
    The trapdoor supports matrix-vector multiplication by $M$, $M^\tp$, $M^{-1}$, and $M^{-\tp}$ in $n^{1+\epsilon}$ time and polylogarithmic depth.
    The time can be reduced into $n \cdot \poly\log (n)$ assuming subexponential hardness.
\end{theorem}

Along the way, we also construct the trapdoored triangular and symmetric matrices. Together all of these, we can embed them into the Bravyi--Maslov Clifford sampler~\cite{bravyi2021hadamard}, which was obtained using the Bruhat decomposition. Interestingly, the Bruhat decomposition is used in constructing the above trapdoored matrices with efficient inversions.
This structure also allows efficient computation of determinant of the sampled matrices, which answers another question in \cite{vaikuntanathan2026improving} affirmatively.

\subsection{Applications}

\paragraph{Immediate applications.}
For many fundamental protocols repeatedly sample and apply random Clifford operations, our construction immediately replace them by the trapdoored Clifford operators.

\begin{itemize}
\item \textit{Efficient simulation and verification.}
Given the trapdoor, conjugating a Pauli operator by Cliffords takes near-linear time.
This speeds up the Pauli calculations in stabilizer simulation~\cite{aaronson2004improved} and direct fidelity estimation of Clifford channels~\cite{FlammiaLiu11}. It also can be used in few $T$-gate simulation~\cite{BravyiGosset16}, imposing random-looking Cliffords to the actual quantum algorithm while allowing efficient classical simulation using trapdoors.

\item \textit{3-designs and classical shadows.}
Uniform random Cliffords form an exact 3-design \cite{Zhu2017CliffordDesign,Webb2016CliffordDesign}. The trapdoored Clifford operators may lose this property, yet they still look like an (approximate) 3-design from the efficient algorithm's view.
This motivates its use in classical shadows~\cite{huang2020predicting}: one applies a sampled Clifford to an unknown state and records the measurement outcome repeatedily.
The resulting data can later be used to predict many efficiently computable observables, including fidelities with stabilizer states chosen after measurement.

\item \textit{Random stabilizer codes.}
Applying a uniform random Clifford to a state followed by zero ancillas gives a random stabilizer code.
This can be used in several protocols including quantum secret sharing and erasure-resilient communications~\cite{PorembaQuekShor26,Matsumoto17,KuoOuyang26}.

\item \textit{Clifford authentication.}
The same template used in
the Clifford authentication~\cite{aharonov2017interactive}.
Using trapdoored Clifford improves their efficiency; for the quantum authentication, this argument only ensures the soundness against efficient attacks.
Adapting the technique used in~\cite{aharonov2017interactive} turning arbitrary adversary into Paulis, we actually prove the negligible soundness error against all adversarial behavior in \cref{subsec: Clifford auth}.
\end{itemize}

Other applications include gate-set shadow estimation and quantum data hiding~\cite{HelsenEtAl23,DiVincenzoLeungTerhal02}.

\paragraph{Worst-case to average-case reductions.}
The worst-to-average reduction for the matrix multiplication asks if given a black-box algorithm to compute the matrix multiplication for an average input,
one can compute $A \cdot B$ for almost all $A,B$ using almost the same time. It was considered in \cite{STOC:BluLubRub90} and answered in a positive way.
Subsequent works showed that this is possible even if the given algorithm works correctly (or even approximately correctly \cite{STOC:HirShi25,GSS24}) only on a very small fraction of inputs \cite{STOC:AGGS22,STOC:HirShi23}.

Using the trapdoored matrices, \cite{vaikuntanathan2026improving} showed that this worst-to-average behavior is pravelent in many linear algebra problems.
Our trapdoored invertible matrices also yield worst-case-to-average-case reductions for linear algebra and Clifford problems, extending this framework.

Our main application for the linear algebra problems are the iterated matrix multiplication and determinant computation. The worst-case to average-case reduction for the iterated matrix multiplication was studied in \cite{STOC:GLRSW91,GS92} for the black-box algorithm works correctly most (or at least half for large characteristic \cite{GS92}) of the inputs. We show the worst-case to average-case reduction for the iterated matrix multiplication even when the given algorithm works correctly for a tiny fraction of the inputs.
The determinant computation was studied in \cite{vaikuntanathan2026improving}. We improve the error rate: For any $\eps>0$, the reduction works with the error rate $1/2+\eps$ for sufficiently large $q$.

Finally, we study the cost of synthesizing Clifford circuits. We prove the worst-to-average reduction for the Clifford circuit synthesis given the tableau. Further, we relate the average-case of \emph{batched} Clifford circuits to the worst-case matrix multiplication.
Suppose an algorithm, given the tableau of a random $2n$-qubit Clifford $C$, constructs a Clifford circuit for $C^{\otimes n}$ with constant success probability.
We convert it into an algorithm that multiplies any two $n\times n$ matrices over $\F_2$ with high probability.
Thus, faster synthesis of circuits applying the same Clifford operator to multiple registers would imply faster worst-case matrix multiplication, even if synthesis produces approximate implementations for only an arbitrarily small constant fraction of random Clifford operators.
In other words, batched Clifford circuit synthesis is at least as hard as matrix multiplication.

\subsection{Notation}\label{Notations}

\noindent\textit{Fields and Matrices.}
Let $q$ be a fixed prime power, and let $\F_q$ denote the field with $q$ elements.
For positive integers $n$ and $m$, write $[n] \coloneqq \{1,\ldots,n\}$ and
let $\Mat_{n\times m}(\F_q)$ denote the set of $n\times m$ matrices over $\F_q$.
All matrix classes below are defined over $\F_q$; we omit the field from the notation.
We write $\Upper_n$ and $\Lower_n$ for the sets of $n\times n$ upper and lower triangular matrices, respectively.
The subsets of strictly upper triangular and invertible upper triangular matrices are denoted by $\strictupper$ and $\Upper_n^*$, respectively:
\begin{align*}
    \strictupper
        &\coloneqq \bigl\{U\in\Upper_n \bigm| U_{ii}=0 \text{ for all } i\in[n]\bigr\},\\
    \Upper_n^*
        &\coloneqq \bigl\{U\in\Upper_n \bigm| U_{ii}\ne 0 \text{ for all } i\in[n]\bigr\}.
\end{align*}
The classes $\strictlower$ and $\Lower_n^*$ are defined analogously.
We denote the set of symmetric matrices by
\begin{align*}
    \Sym_n &\coloneqq \bigl\{M\in\Mat_{n\times n}(\F_q) \bigm| M^\tp=M\bigr\},
\end{align*}
where $M^\tp$ denotes the transpose of $M$.
Finally, $\CS_n$ denotes the set of $n\times n$ permutation matrices, which have exactly
one entry equal to $1$ in each row and each column and all other entries equal to $0$.

\vspace{4mm}

\noindent \textit{Cryptography. } A nonnegative function $\negl(n)$ is negligible if, for every $c > 0$, it is smaller than $n^{-c}$ for all sufficiently large $n$.
We use computational indistinguishability against any quantum polynomial-time (QPT) distinguishers, denoted by $\approx_c$.

\vspace{4mm}

\noindent \textit{Complexity. } All logarithms are to base $2$. All asymptotic statements are with respect to the dimension $n$ as $n \to \infty$.
We write $\ql = \oracle(n \cdot \operatorname{polylog} n)$
for quasi-linear complexity bounds.

\vspace{4mm} 

\noindent \textit{Probability and Distributions. } The uniform distribution over a set $X$ is denoted by $\unif(X)$.
For a noise rate $\varepsilon \in (0,1)$, let $\Ber(\varepsilon)$ denote the distribution over $\F_q$ that outputs $0$ with probability $1-\varepsilon$ and a uniformly random element of $\F_q \setminus \{0\}$ with probability $\varepsilon$.

\section{Trapdoored Matrix Distributions}\label{sec:trapdoored-matrices}

We begin with a general definition of trapdoored matrix distributions. 
We then formalize our (dual) ring-LPN assumptions and review the existing construction that serves as our starting point. 
Building on this construction, we obtain trapdoored triangular matrices and extend the construction to the invertible matrices. 
The resulting distributions admit efficient sampling, and their trapdoors support efficient matrix–vector multiplication by both the sampled matrix and its inverse.

\subsection{Definition}\label{trapdoor def}

\begin{definition}[Trapdoored Matrix Distribution]\label{def:trapdoored-matrix}
    For any $n \in \N$ and $m = \poly(n)$, let $\CH \subseteq \Mat_{n\times m}(\F_q)$ be a nonempty family of matrices and $\mathcal{F}=\{f_i: \CH \times \mathcal{X}_i \rightarrow \mathcal{Y}_i\}$ be a family of mappings. 
    An efficiently samplable distribution $\D$ over tuples
    \[
        (M,\{C_{f_i,M}\}),
    \]
    where $M \in \CH$ and each $C_{f_i,M}$ is a circuit, is called a trapdoored matrix distribution over $\CH$ supporting $\mathcal{F}$ if the following conditions hold:
    \begin{enumerate}[label=(\roman*)]
        \item \textbf{Correctness.}
        For every sampled tuple $(M,\{C_{f_i, M}\})$, every $f_i\in\mathcal{F}$, and every $x\in\mathcal{X}_i$,
        \[
            C_{f_i,M}(x) = f_i(M,x).
        \]

        \item \textbf{Pseudorandomness.}
        The marginal distribution of $M$ under $\D$ is computationally indistinguishable from $\unif(\CH)$.
        Namely, for every QPT distinguisher $\CA$, for all sufficiently large $n$, with $m=\poly(n)$ as above,
        \[
            \left|
            \Pr_{(M,\{C\})\leftarrow\D}\!\left[\CA(M)=1\right]
            -\Pr_{U\leftarrow\unif(\CH)}\!\left[\CA(U)=1\right]
            \right|\le\negl(n).
        \]
    \end{enumerate}
    The circuit family $\{C_{f_i, M}\}$ is called the trapdoor for $M$.
\end{definition}

For any fixed $\CH \subseteq \Mat_{n \times m}(\F_q)$, we define the matrix--vector multiplication map
\[\mult: \CH \times \F_q^m \rightarrow \F_q^n, \qquad \mult(H, \vect{x}) = H \vect{x},\]
and the transpose multiplication map
\[\mult^\tp: \CH \times \F_q^n \rightarrow \F_q^m, \qquad \mult^\tp(H, \vect{x}) = H^\tp \vect{x}.\]

When $n = m$, we also define the inverse multiplication map
\[\inv: \CH \times \F_q^n \rightarrow \F_q^n \cup \{\perp\}, \qquad \inv(H, \vect{x}) = \begin{cases} H^{-1} \vect{x}, & \text{ if $H$ is invertible, } \\ \perp, & \text{ otherwise, }\end{cases}\]
and the inverse-transpose multiplication map
\[\inv^\tp : \CH \times \F_q^n \rightarrow \F_q^n \cup \{\perp\}, \qquad \inv^\tp(H, \vect{x}) = \begin{cases} (H^{\tp})^{-1} \vect{x} = (H^{-1})^\tp \vect{x}, & \text{ if $H$ is invertible, }  \\ \perp, & \text{ otherwise. }\end{cases}\]

We track the sampling complexity of the distribution and the evaluation complexity of its trapdoor circuits, as formalized below.

\begin{definition}[Sampling and Trapdoor Efficiency]\label{def:trapdoor-efficiency}
    Let $\D$ be a trapdoored matrix distribution supporting $\mathcal{F} = \{f_i\}$.
    We say that $\D$ is $T$-samplable if a description of a tuple drawn from $\D$ can be sampled in time $T$ with overwhelming probability.\footnote{Here, overwhelming probability means probability at least $1-\negl(n)$ over the sampling randomness. This qualification concerns efficiency only; correctness holds for every sampled tuple.}
    We say that $\D$ supports $T_i$-efficient evaluation of $f_i$ if, with overwhelming probability, the sampled circuit $C_{f_i,M}$ can be evaluated in time $T_i$ on every input.
\end{definition}

\begin{remark}[Sampling trapdoors and explicit representations]
    Sampling efficiency measures the time needed to produce the
    trapdoor, which also serves as a succinct description of the
    sampled matrix $M$.
    Constructing the full array of entries of $M$ from this
    description incurs an additional cost, accounted for separately.
    In the pseudorandomness experiment, the distinguisher receives
    this explicit matrix representation, without the trapdoor.
    The same convention applies to our trapdoored Clifford operators:
    the sampled trapdoor fully determines the operator, whereas
    the distinguisher receives its signed tableau without the trapdoor.
\end{remark}

\subsection{Base Construction from (Dual) Ring-LPN Assumptions}\label{section2.2.}

The standard decisional learning parity with noise (LPN) assumption concerns noisy linear equations over $\F_2$~\cite{blum1993cryptographic, FOCS:Alekhnovich03, lyubashevsky2013man}; 
we use its analogue over a fixed finite field $\F_q$~\cite{CCS:BCGI18}.
For a suitable noise rate $\eps > 0$, this assumption asserts that polynomially many samples
\[(\vect{r}_i, \vect{r}_i \cdot \vect{s} + e_i), \quad \vect{r}_i \leftarrow \unif(\F_q^n), \, e_i \sim \Ber(\eps),\]
with a shared secret $\vect{s} \leftarrow \unif(\F_q^n)$ are jointly computationally indistinguishable from independent uniform samples in $\F_q^n \times \F_q$.
The vectors $\vect{r}_i$ and errors $e_i$ are sampled independently of one another and of $\vect{s}$.

Ring-LPN is a structured variant of LPN in which the noisy linear equations are defined over a polynomial quotient ring~\cite{heyse2012lapin,boyle2020efficient}.
We consider the ring $R=\F_q[x]/(x^n-1)$ and sample a shared secret $s\leftarrow\unif(R)$.
For a suitable noise rate $\eps>0$, each error is a polynomial
\[e_i(x)=\sum_{j=0}^{n-1}e_{ij}x^j, \qquad e_{ij}\sim\Ber(\eps).\]
The corresponding decisional ring-LPN assumption asserts that polynomially many samples
\[(r_i,r_i s+e_i), \qquad r_i\leftarrow\unif(R),\]
are jointly computationally indistinguishable from independent uniform samples in $R\times R$.
All $r_i$ and error coefficients $e_{ij}$ are sampled independently of one another and of $s$, and all operations are performed in $R$.

A dual-LPN assumption is specified by an efficient sampler for a public matrix $H \in \Mat_{n \times k}(\F_q)$, where $n<k=\poly(n)$.
For a specified matrix distribution and suitable noise rate $\eps > 0$, the dual-LPN assumption states that~\cite{benhamouda2025encrypted}
\[(H, \vect{u}) \approx_c (H, H\vect{e}), \quad \vect{u} \leftarrow \unif(\F_q^n), \, \vect{e} \sim \Ber(\eps)^k. \]
The sampling of $H$ is included in both distributions, and $H$, $\vect{e}$, and $\vect{u}$ are sampled independently.
A standard hybrid argument over the columns then implies that, for any $m=\poly(n)$ and the same sampled $H$,
\[(H, U) \approx_c (H, HE), \quad U \leftarrow \unif(\Mat_{n \times m}(\F_q)), \, E \sim \Ber(\eps)^{k \times m}.\]

Structured LPN variants use public matrices that admit fast matrix--vector multiplication, such as Toeplitz matrices~\cite{lyubashevsky2013man} and suitable quasi-cyclic matrices~\cite{boyle2020efficient,aguilar2018efficient}.
Following~\cite{benhamouda2025encrypted}, we use quasi-cyclic matrices of the form $H = [I_n \mid C]$, where $C$ is a uniformly random $n\times n$ circulant matrix.
Identifying vectors with coefficient vectors in $R$, multiplication by $C$ corresponds to multiplication by a uniformly random ring element $a\in R$.
Thus, writing $\vect{e}=(\vect{e}_0,\vect{e}_1)$, the syndrome $H\vect{e}$ corresponds to $e_0+ae_1$, where both $e_0$ and $e_1$ have independent noisy coefficients.
We refer to the hardness assumption for this distribution as \emph{(dual) ring-LPN}; its secret $e_1$ is sparse, whereas the introductory ring-LPN formulation above uses a uniform secret.
Multiplication by $H$ takes $\oracle(n\log n\log\log n)$ field operations via fast polynomial multiplication~\cite{cantor1991fast}, and hence admits a circuit of size $\oracle(n\log^2 n)$.

We formalize this instantiation in two forms, which we call the \emph{standard (dual) ring-LPN assumption} and the \emph{strong (dual) ring-LPN assumption}.
In each experiment below, the displayed random variables are sampled independently,
and the probability also includes the distinguisher's internal randomness and measurement outcomes.
Let $\CQ_n$ denote the following family of quasi-cyclic matrices:
\[\CQ_n \coloneqq \big\{[I_n \mid C] \mid C \text{ is an $n \times n$ circulant matrix}\big\}.\]

\begin{conjecture}[Standard Ring-LPN Assumption in Dual Form~\cite{benhamouda2025encrypted}]\label{standard assumption}
    For a suitable noise rate $\eps(n)>0$ satisfying
    \[
        \frac{n\eps(n)}{\log^2{n}} \longrightarrow \infty, \quad \eps(n) < \frac{1}{4},
    \]
    every QPT distinguisher $\CA$ satisfies
    \[
        \left|
        \Pr_{\substack{Q \leftarrow \unif(\CQ_n) \\\vect{e}\leftarrow\Ber(\eps(n))^{2n}}}
        [\CA(Q,Q\vect{e})=1]
        -\Pr_{\substack{Q\leftarrow\unif(\CQ_n)\\\vect{u}\leftarrow\unif(\F_q^n)}}
        [\CA(Q,\vect{u})=1]
        \right|\leq\negl(n)
    \]
    for all sufficiently large $n$. The negligible function may depend on $\CA$ and $\eps$. 
\end{conjecture}

In particular, we conjecture that this assumption holds for $\eps(n)=(\log n)^c/n$ with any fixed $c>2$; these rates satisfy the stated parameter restrictions for all sufficiently large $n$.
Motivated by the subexponential LPN assumption in~\cite{vaikuntanathan2026improving},
we also consider a quantitative assumption against quantum distinguishers.

\begin{conjecture}[Strong Ring-LPN Assumption in Dual Form]\label{strong assumption}
    There exists a constant $\alpha\in(0,1)$ such that, for any suitable noise rate $\eps(n)>0$ satisfying
    \[
        \frac{n\eps(n)}{\log^2{n}} \longrightarrow \infty, \quad \eps(n) < \frac{1}{4},
    \]
    every quantum $2^{(n\eps(n))^\alpha}$-time distinguisher $\CA$ satisfies
    \[
        \left|
        \Pr_{\substack{Q \leftarrow \unif(\CQ_n) \\\vect{e}\leftarrow\Ber(\eps(n))^{2n}}}
        [\CA(Q,Q\vect{e})=1]
        -\Pr_{\substack{Q\leftarrow\unif(\CQ_n)\\\vect{u}\leftarrow\unif(\F_q^n)}}
        [\CA(Q,\vect{u})=1]
        \right|\leq 2^{-(n\eps(n))^\alpha}
    \]
    for all sufficiently large $n$.
\end{conjecture}

The strong assumption quantifies security as a function of the noise rate on the stated parameter range.
For $\eps(n)=(\log n)^c/n$ with $c>2$ and $c\alpha>1$, it guarantees negligible advantage against every QPT distinguisher.
For any fixed noise rate $\eps\in(0,1/4)$, it assumes security against
$2^{(n\eps)^\alpha}=2^{\Theta(n^\alpha)}$-time quantum algorithms, with advantage at most $2^{-(n\eps)^\alpha}=2^{-\Theta(n^\alpha)}$.

Several constructions of trapdoored matrix distributions over $\Mat_{n \times m}(\F_q)$ are known under variants of the LPN assumption~\cite{braverman2025practical, benhamouda2025encrypted, vaikuntanathan2026improving}.
We use the (dual) ring-LPN construction of~\cite{benhamouda2025encrypted}, instantiated with the quasi-cyclic matrix family $\CQ_n$.
The following lemma states its guarantees in terms of our efficiency definition.

\begin{lemma}[\cite{benhamouda2025encrypted}]\label{Base TDM}
    Let $n\in\N$, $m=\Theta(n)$, and $\eps(n)=(\log n)^c/n$ for a fixed constant $c>2$.
    Under \cref{standard assumption} at noise rate $\eps(n)$, there exists an $\oracle(nm\eps)$-samplable trapdoored matrix distribution over $\Mat_{n\times m}(\F_q)$
    supporting $\oracle(nm\eps)$-efficient evaluation of $\mult$.
\end{lemma}

\begin{proof}
    Independently sample $H=[I_n\mid C]\leftarrow\unif(\CQ_n)$ and $E\sim\Ber(\eps)^{2n\times m}$, and set $M=HE$.
    The trapdoor evaluates $M\vect{x}$ as $H(E\vect{x})$, so correctness is immediate.
    Pseudorandomness follows from \cref{standard assumption} by the column hybrid argument above.

    Let $W$ denote the number of nonzero entries of $E$, so that $\mathbb{E}[W]=2nm\eps$.
    Sparse multiplication by $E$ costs $\oracle(n+m+W)$, while quasi-cyclic multiplication by $H$ costs $\oracle(n\log^2 n)$.
    Thus the expected circuit size is $\oracle(nm\eps+n\log^2 n)$, as in~\cite[Proposition~5.1]{benhamouda2025encrypted}.
    For sampling, we store only the $n$ defining coefficients of $C$ and the nonzero entries of $E$.
    Exact geometric skipping~\cite{ICALP:BriFri13} generates these sparse entries in expected time $\oracle(nm\eps)$, giving the same expected bound for sampling the succinct matrix description and its trapdoor circuit.

    By a Chernoff bound, $W=\oracle(nm\eps)$ with overwhelming probability, since $nm\eps=\Theta(n(\log n)^c)$.
    Hence the sampled circuit runs in $\oracle(nm\eps+n\log^2 n)$ time on every input with overwhelming probability.
    Concentration of the gap-sampling times gives the same bound for sampling.\footnote{We use a randomized word RAM with $\Theta(\log n)$-bit words and polynomial-time precision access to the normalized noise rate. The exact sampler of~\cite{ICALP:BriFri13} has uniformly bounded running-time moments of every fixed order; a standard moment bound on $\oracle(nm\eps)$ independent calls gives the claimed overwhelming-probability runtime.}

    Finally, $m=\Theta(n)$ and $c>2$ imply $n\log^2 n=o(nm\eps)$, so the sparse-error contribution dominates the quasi-cyclic cost:
    \[
        \oracle(nm\eps+n\log^2 n)=\oracle(nm\eps).
    \]
    Both efficiency requirements of \cref{def:trapdoor-efficiency} follow.
\end{proof}

\subsection{Trapdoored Triangular Matrices}\label{section2.3.}

For $n, m \in \N$ with $m = \Theta(n)$ and a suitable noise rate $\eps(n) > 0$, let $\D_\CM(n, m, \eps)$ denote 
the trapdoored matrices distribution over $\Mat_{n \times m}(\F_q)$ defined in \cref{Base TDM}. 
In this subsection, we construct trapdoored matrix distributions over triangular matrices with efficient sampling and trapdoor circuits for $\mult, \mult^\tp$. 
The complexity bounds are near-linear under the standard (dual) ring-LPN assumption and quasi-linear under its strong variant.
\cref{cor:triangular-matrices} summarizes the results for upper-triangular, lower-triangular, and symmetric matrices. 
We begin with strictly upper-triangular matrices. 

\begin{theorem}\label{thm:strict-upper}
    Recall that $\strictupper$ denotes the set of $n \times n$ strictly upper-triangular matrices. 
    The following statements hold: 
    \begin{enumerate}[label=(\roman*)]
        \item Under the standard (dual) ring-LPN assumption, for any $0 < \delta < 1$, there exists an $\nl$-samplable
        trapdoored matrix distribution over $\strictupper$ supporting $\nl$-efficient $\mult$. 
        \item Under the strong (dual) ring-LPN assumption, there exists an $\ql$-samplable 
        trapdoored matrix distribution over $\strictupper$ supporting $\ql$-efficient $\mult$. 
    \end{enumerate}
\end{theorem}

Both parts of \cref{thm:strict-upper} use the same sampling algorithm, \cref{alg:strict-upper} below.
Fix an integer $c>2$; in the strong case, choose $c$ so that $c\alpha>1$,
where $\alpha$ is the constant in \cref{strong assumption}.
The choice of assumption determines the base dimension $b$.
For part~\textup{(i)}, we take $b = \lceil n^{\delta} \rceil$ for the chosen $0 < \delta < 1$; for part~\textup{(ii)}, the strong assumption allows us to reduce $b$ to $\lceil 8\log^c{n} \rceil$. 
These choices yield the two complexity bounds in \cref{thm:strict-upper}. 
We present the common sampling algorithm and establish its pseudorandomness and efficiency for both choices of $b$. 

\begin{algorithm}[H]
\caption{Sampling from a Trapdoored Matrix Distribution over $\strictupper$}
\label{alg:strict-upper}
\begin{algorithmic}[1]
    \Require Initial dimension $n$, an integer $c>2$, and a base dimension $b$, chosen as above for the relevant assumption.
    \Procedure{StrictUpperTriangular}{$d$}
    \If{$d \leq b$}
    \State $U \leftarrow \unif(\Upper_{d}^\circ)$
    \State Let $C_{\mult, U}$ be a $\oracle(d^2)$-circuit computing $\vect{x} \mapsto U\vect{x}$ for any $\vect{x} \in \F_q^{d}$. 
    \Else
    \State $(M, C_{\mult, M}) \leftarrow \D_\CM(\lceil d/2 \rceil, \lfloor d/2 \rfloor, (2\log^c{n})/d)$
    \State $(U_1, C_{\mult, U_1}) \leftarrow \Call{StrictUpperTriangular}{\lceil d/2 \rceil}$ 
    \State $(U_2, C_{\mult, U_2}) \leftarrow \Call{StrictUpperTriangular}{\lfloor d/2 \rfloor}$
    \State $U \leftarrow \begin{bmatrix} U_1 & M \\ \vect{0} & U_2 \end{bmatrix}$ 
    \State $C_{\mult, U}: (\vect{x}_1, \vect{x}_2) \mapsto (C_{\mult, U_1}(\vect{x}_1) + C_{\mult, M}(\vect{x}_2), C_{\mult, U_2}(\vect{x}_2))$ for any $\vect{x}_1 \in \F_q^{\lceil d/2 \rceil}, \vect{x}_2 \in \F_q^{\lfloor d/2 \rfloor}$.  
    \EndIf
    \State \Return $(U, C_{\mult, U})$
    \EndProcedure
\end{algorithmic}
\end{algorithm}

The following three lemmas establish pseudorandomness, sampling efficiency, and circuit size, respectively.
Throughout the remainder of this subsection, we assume that the initial dimension $n$ is a sufficiently large power of $2$
and that $c$ is chosen as specified above.

\begin{lemma}\label{Pseudorandomness}
    The matrix distribution induced by \cref{alg:strict-upper} is computationally indistinguishable from $\unif(\strictupper)$
    under either of the following parameter choices:
    under the standard (dual) ring-LPN assumption, take $b=\lceil n^\delta\rceil$ for a fixed $0<\delta<1$ and an integer $c>2$;
    under the strong (dual) ring-LPN assumption, take $b=\lceil 8\log^c n\rceil$ with an integer $c>2$ satisfying $c\alpha>1$.
\end{lemma}

\begin{proof}
    Let $\mathsf{Hybrid}_i$ denote the matrix distribution obtained by setting the base dimension in \cref{alg:strict-upper} to
    \[
        b_i=\max\left\{\frac{n}{2^i},b\right\},
    \]
    while keeping the original $n$ in the noise rates.
    Then $\mathsf{Hybrid}_0=\unif(\strictupper)$, and $\mathsf{Hybrid}_L$ coincides with the output distribution of \cref{alg:strict-upper}, where $L=\lceil\log(n/b)\rceil$.
    Fix a QPT distinguisher $\CA$, and write
    \[
        d_i=\frac{n}{2^i},\qquad
        \eps_i=\frac{2\log^c n}{d_i}=\frac{2^{i+1}\log^c n}{n}.
    \]

    For $i\in[L]$, adjacent hybrids differ only in $2^{i-1}$ off-diagonal blocks: these are uniform in $\mathsf{Hybrid}_{i-1}$ and sampled from $\D_\CM(d_i,d_i,\eps_{i-1})$ in $\mathsf{Hybrid}_i$.
    A hybrid argument over the blocks and then over their $d_i$ columns reduces distinguishing the two hybrids to distinguishing a uniform vector from $Q\vect{e}$, where $Q\leftarrow\unif(\CQ_{d_i})$ and $\vect{e}\sim\Ber(\eps_{i-1})^{2d_i}$.
    The columns within each real block share the same public $Q$.
    Choosing a block and a column uniformly gives a single-sample quantum reduction running in time $\poly(n)$.
    Denoting its advantage by $\eta_i$, we obtain
    \begin{align*}
        \Delta_i
        &\coloneqq\left|\Pr[\CA(\mathsf{Hybrid}_i)=1]-\Pr[\CA(\mathsf{Hybrid}_{i-1})=1]\right|\\
        &\leq 2^{i-1}d_i\eta_i=\frac{n}{2}\eta_i.
    \end{align*}

    Under the standard (dual) ring-LPN assumption in \cref{standard assumption}, take $b=\lceil n^\delta\rceil$ for a fixed $0<\delta<1$.
    Since $d_i\geq d_L>b/2$, for sufficiently large $n$ we have
    \[
        d_i>\frac{n^\delta}{2},\qquad
        \eps_{i-1}=\frac{\log^c n}{d_i}<\frac{2\log^c n}{n^\delta}<\frac14.
    \]
    Moreover,
    \[
        \frac{d_i\eps_{i-1}}{\log^2 d_i}
        =\frac{\log^c n}{\log^2 d_i}
        \geq\log^{c-2}n\longrightarrow\infty.
    \]
    Thus the noise rates lie in the required range.
    Since $n<(2d_i)^{1/\delta}$, each reduction runs in polynomial time in $d_i$.
    The assumption therefore gives $\eta_i=\negl(d_i)=\negl(n)$, and summing over the $L=\oracle(\log n)$ levels yields
    \[
        \sum_{i=1}^L\Delta_i\leq\frac{n}{2}\sum_{i=1}^L\eta_i=\negl(n).
    \]

    Under the strong (dual) ring-LPN assumption in \cref{strong assumption}, take $b=\lceil8\log^c n\rceil$.
    Now $d_i>b/2\geq4\log^c n$, so
    \[
        \eps_{i-1}=\frac{\log^c n}{d_i}<\frac14,\qquad
        \frac{d_i\eps_{i-1}}{\log^2 d_i}\geq\log^{c-2}n\longrightarrow\infty.
    \]
    Hence these noise rates also lie in the required range.
    Since $c\alpha>1$, the time bound
    \[
        2^{(d_i\eps_{i-1})^\alpha}
        =2^{(\log n)^{c\alpha}}
        =n^{(\log n)^{c\alpha-1}}
    \]
    eventually exceeds the polynomial running time of each reduction.
    The single-sample advantage is therefore bounded by $\eta_i\leq2^{-(\log n)^{c\alpha}}$, giving
    \[
        \sum_{i=1}^L\Delta_i
        \leq\frac{nL}{2}\,2^{-(\log n)^{c\alpha}}
        =\negl(n).
    \]
    In either case, $\mathsf{Hybrid}_L$ is computationally indistinguishable from $\mathsf{Hybrid}_0$, completing the proof.
\end{proof}

\begin{lemma}\label{lem:strict-upper-sampling}
    \cref{alg:strict-upper} runs in time $\oracle(n\log^{c+1}{n} + nb)$ with overwhelming probability. 
\end{lemma}

\begin{proof}
    Let $L=\lceil\log(n/b)\rceil$.
    At depth $i=0,\ldots,L-1$, there are $2^i$ blocks of dimension $d_i=n/2^i$, each using an error matrix of size $d_i\times(d_i/2)$ at rate $\eps_i=2\log^c n/d_i$.
    Let $W_i$ be the total number of nonzero error entries at this depth. Then
    \[
        \mathbb{E}[W_i]=2^i\cdot\frac{d_i^2}{2}\cdot\eps_i=n\log^c n.
    \]
    Chernoff bounds and a union bound over the $L=\oracle(\log n)$ depths give $W_i=\oracle(n\log^c n)$ simultaneously at every depth with overwhelming probability in $n$.

    The error matrices are sampled by exact geometric skipping.
    Applying the sampling-time concentration argument from \cref{Base TDM} to the gap calls across all blocks at each depth gives $\oracle(n\log^c n)$ sampling time with overwhelming probability.
    Sampling the defining coefficients and constructing the circuits adds $\oracle(n\log^2 n+W_i)$ time with overwhelming probability.
    Since $c>2$, the total cost per depth is $\oracle(n\log^c n)$.

    The leaf dimensions sum to $n$ and are at most $b$, so sampling the uniform leaf matrices and constructing their circuits costs $\oracle(nb)$ with overwhelming probability.
    Summing these costs gives
    \[
        \oracle\bigl(n\log^c n\cdot\log(n/b)+nb\bigr)
        =\oracle(n\log^{c+1}n+nb),
    \]
    as claimed.
\end{proof}

\begin{lemma}\label{lem:strict-upper-circuit-size}
    With overwhelming probability, the circuit $C_{\mult,U}$ sampled by \cref{alg:strict-upper} has size $\oracle(n\log^{c+1}n+nb)$.
\end{lemma}

\begin{proof}
    Let $W_i$ be the total error weight at depth $i$, as in the proof of \cref{lem:strict-upper-sampling}.
    Sparse multiplication contributes $\oracle(n+W_i)$ gates, while quasi-cyclic multiplication and the additions combining the recursive blocks contribute $\oracle(n\log^2 n)$ gates.
    With overwhelming probability, $W_i=\oracle(n\log^c n)$ at every depth, so $c>2$ gives a total of $\oracle(n\log^c n)$ gates per depth.
    Summing over the $\oracle(\log(n/b))$ depths and adding the $\oracle(nb)$ gates at the leaves yields $\oracle(n\log^{c+1}n+nb)$, completing the proof.
\end{proof}

The distribution sampled by \cref{alg:strict-upper} also supports efficient $\mult^\tp$.
We obtain the corresponding trapdoor circuit by applying Tellegen's principle, also known as the transposition principle.
For later use, we state the principle in the following general form.

\begin{lemma}[Tellegen's Principle~\cite{bostan2003tellegen}]\label{lem:Tellegen}
    Let $\F$ be a field and $M\in\Mat_{n\times m}(\F)$, where $n,m\in\N$ and $m=\oracle(n)$.
    Suppose a circuit $C$ of size $T(n)$ computes $\vect{x}\mapsto M\vect{x}$ using only additions and multiplications by constants in $\F$.
    Then a circuit $C^\tp$ of size $\oracle(T(n)+n)$ computing $\vect{y}\mapsto M^\tp\vect{y}$ can be constructed from $C$ in time $\oracle(T(n)+n)$.
\end{lemma}

\begin{remark}
    Once the sampled trapdoor is fixed, all circuits we construct for $\mult$ and for $\inv$ on invertible matrices use only additions and multiplications by field constants.
    Unrolling any recursive calls therefore yields linear arithmetic circuits to which \cref{lem:Tellegen} applies.
    Throughout the paper, we will invoke this lemma without further comment to obtain the corresponding circuits for $\mult^\tp$ and $\inv^\tp$.
    Our sampling and circuit-size bounds have the same order and are at least linear in $n$, so this preserves the stated sampling and evaluation bounds.
\end{remark}

Combining \cref{thm:strict-upper} with \cref{lem:Tellegen}, we extend the construction to the matrix families in the following corollary. 
Recall that $\Lower_n^{0}$ and $\Lower_n$ denote the sets of strictly lower-triangular and lower-triangular $n \times n$ matrices over $\F_q$, respectively, and 
$\Sym_n$ denotes the set of $n \times n$ symmetric matrices over $\F_q$. 

\begin{corollary}\label{cor:triangular-matrices}
    Let $\CH$ be any of the matrix families $\Upper_n^\circ$, $\Lower_n^\circ$, $\Upper_n$, $\Lower_n$, or $\Sym_n$.
    Then the following statements hold:
    \begin{enumerate}[label=(\roman*)]
        \item Under the standard (dual) ring-LPN assumption, for any fixed $0<\delta<1$, there exists an $\nl$-samplable trapdoored matrix distribution over $\CH$ supporting $\nl$-efficient evaluation of both $\mult$ and $\mult^\tp$.
        \item Under the strong (dual) ring-LPN assumption, there exists an $\ql$-samplable trapdoored matrix distribution over $\CH$ supporting $\ql$-efficient evaluation of both $\mult$ and $\mult^\tp$.
    \end{enumerate}
\end{corollary}

\begin{proof}
    Sample $U$ from the distribution in \cref{thm:strict-upper}, and independently sample a diagonal matrix $D$ whose diagonal entries are independent and uniform in $\F_q$.
    For the five choices of $\CH$ listed above, use, respectively,
    \[
        U,\qquad U^\tp,\qquad U+D,\qquad U^\tp+D,\qquad U+U^\tp+D.
    \]
    When $U$ is uniform over $\strictupper$, these matrices are uniform over their respective families.
    Pseudorandomness therefore follows by efficient post-processing.

    Each multiplication circuit uses at most one copy of the circuit for $U$ and one for $U^\tp$, together with $\oracle(n)$ additional gates for diagonal multiplication and additions.
    The same holds for transpose multiplication.
    Sampling $D$ and constructing these circuits preserves the stated sampling and evaluation bounds, completing the proof.
\end{proof}

\subsection{Trapdoored Invertible Triangular Matrices with Efficient Inversion}\label{section2.4.}

We now extend the construction in \cref{cor:triangular-matrices} to $\Upper_n^*$, the group of invertible upper-triangular matrices over $\F_q$ (a Borel subgroup of $\GL(n,\F_q)$).
In addition to matrix--vector multiplication, the resulting trapdoors support the inverse and inverse-transpose actions $\inv$ and $\inv^\tp$ defined in \cref{trapdoor def}.
The following theorem summarizes our guarantees.

\begin{theorem}\label{Trapdoored Invertible Triangular}
    Let $\Upper_n^*$ denote the group of $n \times n$ invertible upper-triangular matrices over $\F_q$.
    \begin{enumerate}[label=(\roman*)]
        \item Under the standard (dual) ring-LPN assumption, for any fixed $0<\delta<1$, there exists an $\nl$-samplable trapdoored matrix distribution over $\Upper_n^*$ supporting $\nl$-efficient $\mult$, $\mult^\tp$, $\inv$, and $\inv^\tp$.
        \item Under the strong (dual) ring-LPN assumption, there exists an $\ql$-samplable trapdoored matrix distribution over $\Upper_n^*$ supporting $\ql$-efficient $\mult$, $\mult^\tp$, $\inv$, and $\inv^\tp$.
    \end{enumerate}
\end{theorem}

Both parts of \cref{Trapdoored Invertible Triangular} use \cref{alg:invertible-upper}, with \cref{alg:efficient-inversion} as a subroutine.
As in the preceding subsection, fix an integer $c>2$; in the strong case, also require $c\alpha>1$, where $\alpha$ is the constant in \cref{strong assumption}.
We take $b=\lceil n^\delta\rceil$ in part~\textup{(i)} and $b=\lceil8\log^c n\rceil$ in part~\textup{(ii)}.

An upper-triangular matrix is invertible if and only if all its diagonal entries are nonzero.
Accordingly, \cref{alg:invertible-upper} forms $U^*=U+D$, where $U$ is sampled by \cref{alg:strict-upper} and $D$ is an independent, uniformly random invertible diagonal matrix.
\cref{alg:efficient-inversion} follows the recursive block structure of $U$ to construct a circuit $C_{\inv,U^*}$ computing $\vect{y}\mapsto(U^*)^{-1}\vect{y}$.
The circuit computes the lower diagonal block first and substitutes the result into the upper block; and use back-substitution\footnote{For an invertible upper-triangular matrix $A=(a_{ij})$, back-substitution solves $A\vect{x}=\vect{y}$ by computing
    \[
        x_i=a_{ii}^{-1}\Bigl(y_i-\sum_{j=i+1}^{d}a_{ij}x_j\Bigr),\qquad i=d,d-1,\ldots,1.
    \]
    The diagonal inverses are precomputed circuit constants. There are $d(d-1)/2$ off-diagonal terms, so constructing and evaluating the circuit both take $\oracle(d^2)$ time.}
at the base dimension. 

\begin{algorithm}
\caption{Generating an Efficient Inversion Circuit}
\label{alg:efficient-inversion}
\begin{algorithmic}[1]
    \Require $(U,C_{\mult,U})$ from \cref{alg:strict-upper}, an invertible diagonal matrix $D$, and the same base dimension $b$.
    \Procedure{EfficientInversionCircuit}{$d, U, C_{\mult, U}, D$}
    \State $U^* \leftarrow U+D$
    \If{$d \leq b$}
    \State Construct a size-$\oracle(d^2)$ circuit $C_{\inv,U^*}$ by back-substitution.
    \Else
    \State $d_1 \leftarrow \lceil d/2\rceil$, $d_2 \leftarrow \lfloor d/2\rfloor$
    \State Obtain the stored block descriptions $U_1,U_2,M$ from $U$, and write $D=\operatorname{diag}(D_1,D_2)$.
    \State Obtain the circuits $C_{\mult,U_1}$, $C_{\mult,U_2}$, and $C_{\mult,M}$ from the recursive description of $C_{\mult,U}$.
    \State $C_1 \leftarrow \Call{EfficientInversionCircuit}{d_1,U_1,C_{\mult,U_1},D_1}$
    \State $C_2 \leftarrow \Call{EfficientInversionCircuit}{d_2,U_2,C_{\mult,U_2},D_2}$
    \State Define $C_{\inv,U^*}$ on input $(\vect{y}_1,\vect{y}_2)$ by:
    \Statex \hspace{4.5em} $\vect{z}\leftarrow C_2(\vect{y}_2)$;
    \Statex \hspace{4.5em} output $\bigl(C_1(\vect{y}_1-C_{\mult,M}(\vect{z})),\vect{z}\bigr)$.
    \EndIf
    \State \Return $C_{\inv,U^*}$
    \EndProcedure
\end{algorithmic}
\end{algorithm}

\begin{algorithm}
\caption{Sampling from a Trapdoored Matrix Distribution over $\Upper_n^*$}
\label{alg:invertible-upper}
\begin{algorithmic}[1]
    \Require Dimension $n$, with $c$ and $b$ chosen as in \cref{alg:strict-upper} for the relevant assumption.
    \Procedure{InvertibleUpperTriangular}{$n$}
    \State Sample $D=\operatorname{diag}(a_1,\ldots,a_n)$ with independent $a_i\leftarrow\unif(\F_q\setminus\{0\})$.
    \State Construct a size-$\oracle(n)$ circuit $C_{\mult,D}$ computing $\vect{x}\mapsto D\vect{x}$.
    \State $(U,C_{\mult,U}) \leftarrow \Call{StrictUpperTriangular}{n}$
    \State $U^* \leftarrow U+D$
    \State $C_{\mult,U^*}:\vect{x}\mapsto C_{\mult,U}(\vect{x})+C_{\mult,D}(\vect{x})$
    \State $C_{\inv,U^*} \leftarrow \Call{EfficientInversionCircuit}{n,U,C_{\mult,U},D}$
    \State \Return $(U^*,C_{\mult,U^*},C_{\inv,U^*})$
    \EndProcedure
\end{algorithmic}
\end{algorithm}

Pseudorandomness follows from \cref{Pseudorandomness} by common hybrid argument. 
The sampled multiplication circuit $C_{\mult, U^*}$ has the same asymptotic size as in \cref{cor:triangular-matrices}, since the diagonal contribution costs only $\oracle(n)$.
It remains to establish the correctness of \cref{alg:efficient-inversion} and to bound its running time and the size of its output circuit $C_{\inv,U^*}$.

\begin{lemma}\label{lem:inversion-runtime}
    \cref{alg:efficient-inversion} runs in time $\oracle(n\log^{c+1}n+nb)$ with overwhelming probability over the input sampled by \cref{alg:strict-upper}, for any invertible diagonal matrix $D$.
    The same running-time bound holds for \cref{alg:invertible-upper} with overwhelming probability.
\end{lemma}

\begin{proof}
    We access the input subcircuits through their stored references and construct the output circuit by writing each gate once.
    At an internal node of dimension $d$, we copy $C_{\mult,M}$ and add $\oracle(d)$ gates to connect the two recursively constructed inverse circuits.
    The child circuits are constructed directly in the output, so their gates are not copied again at their parent nodes.

    By the proof of \cref{lem:strict-upper-circuit-size}, the total size of the off-diagonal circuits at each depth is $\oracle(n\log^c n)$ with overwhelming probability, simultaneously over all depths.
    The connecting gates add only $\oracle(n)$ per depth.
    There are $\oracle(\log(n/b))$ internal levels.
    At the leaves, constructing the circuits costs $\sum_{\ell}\oracle(d_\ell^2)\leq\oracle(nb)$, since $d_\ell\leq b$ and $\sum_\ell d_\ell=n$.
    Thus the total construction time is
    \[
        \oracle\bigl(n\log^c n\cdot\log(n/b)+nb\bigr)
        =\oracle(n\log^{c+1}n+nb).
    \]
    For \cref{alg:invertible-upper}, \cref{lem:strict-upper-sampling,lem:strict-upper-circuit-size} give the same bound for sampling $U$ and constructing the multiplication circuit.
    Sampling $D$ and handling diagonal multiplication add only $\oracle(n)$ work with overwhelming probability, so the claimed bound follows.
\end{proof}

\begin{lemma}\label{lem:inverse-circuit}
    For every input satisfying the requirements of \cref{alg:efficient-inversion}, the constructed circuit $C_{\inv,U^*}$ correctly computes $\vect{y}\mapsto(U^*)^{-1}\vect{y}$.
    With overwhelming probability over the input sampled by \cref{alg:strict-upper}, its size is $\oracle(n\log^{c+1}n+nb)$.
\end{lemma}

\begin{proof}
    Correctness follows by induction on the block dimension $d$.
    For $d\leq b$, the circuit solves $(U+D)\vect{x}=\vect{y}$ by back-substitution.
    Otherwise, let $A_1=U_1+D_1$ and $A_2=U_2+D_2$, both of which are invertible.
    By induction, the algorithm's circuit computes
    \[
        \vect{z}=A_2^{-1}\vect{y}_2,\qquad
        \vect{x}_1=A_1^{-1}(\vect{y}_1-M\vect{z}).
    \]
    Hence its output satisfies
    \[
        (U+D)\begin{bmatrix}\vect{x}_1\\\vect{z}\end{bmatrix}
        =\begin{bmatrix}A_1\vect{x}_1+M\vect{z}\\A_2\vect{z}\end{bmatrix}
        =\begin{bmatrix}\vect{y}_1\\\vect{y}_2\end{bmatrix},
    \]
    proving correctness.

    Each internal node contributes one off-diagonal multiplication circuit and $\oracle(d)$ additional gates, while each leaf contributes $\oracle(d^2)$ gates.
    The gate count in the proof of \cref{lem:inversion-runtime} therefore gives the claimed size bound with overwhelming probability.
\end{proof}

Substituting the two choices of $b$ into \cref{lem:inversion-runtime,lem:inverse-circuit} gives bounds of $\oracle(n^{1+\delta})$ and $\oracle(n\log^{c+1}n)$, respectively.
For each fixed sample, the inverse circuit uses only additions and multiplication by field constants, so \cref{lem:Tellegen} also yields an inverse-transpose circuit with the same asymptotic bounds.
Together with the pseudorandomness and multiplication guarantees above, this proves \cref{Trapdoored Invertible Triangular}.

As in \cref{cor:triangular-matrices}, transposing the sampled matrices gives the same guarantees for the family $\Lower_n^*$ of invertible lower-triangular matrices.
Transposition simply interchanges $\mult$ with $\mult^\tp$ and $\inv$ with $\inv^\tp$, preserving all sampling and evaluation bounds.

\subsection{Trapdoored Invertible Matrices with Efficient Inversion}\label{2.5}

We now extend our construction from invertible triangular matrices to the full general linear group $\GL(n,\F_q)$, preserving efficient sampling and all four trapdoor operations.
The following theorem summarizes the resulting guarantees.

\begin{theorem}\label{trapdoor GL}
    Let $\GL(n,\F_q)$ denote the group of $n\times n$ invertible matrices over $\F_q$.
    \begin{enumerate}[label=(\roman*)]
        \item Under the standard (dual) ring-LPN assumption, for any fixed $0<\delta<1$, there exists an $\nl$-samplable trapdoored matrix distribution over $\GL(n,\F_q)$ supporting $\nl$-efficient $\mult$, $\mult^\tp$, $\inv$, and $\inv^\tp$.
        \item Under the strong (dual) ring-LPN assumption, there exists an $\ql$-samplable trapdoored matrix distribution over $\GL(n,\F_q)$ supporting $\ql$-efficient $\mult$, $\mult^\tp$, $\inv$, and $\inv^\tp$.
    \end{enumerate}
\end{theorem}

Our construction relies on the Bruhat decomposition for split reductive groups.
In this decomposition, the double cosets of a Borel subgroup partition the group and are indexed by elements of the Weyl group.
For $\GL(n,\F_q)$, we take the Borel subgroup to be $\Upper_n^*$ and represent the Weyl group by the permutation matrices $\CS_n$.
This gives the following standard decomposition.

\begin{lemma}[Bruhat decomposition]
    The group $\GL(n,\F_q)$ is the disjoint union of the double cosets $\Upper_n^*\sigma\Upper_n^*$ indexed by $\sigma\in\CS_n$, namely,
    \[
        \GL(n,\F_q)=\bigsqcup_{\sigma\in\CS_n}\Upper_n^*\sigma\Upper_n^*.
    \]
\end{lemma}

To sample uniformly from $\GL(n,\F_q)$ via this decomposition, we first choose $\sigma\in\CS_n$ with probability proportional to the size of its double coset $\Upper_n^*\sigma\Upper_n^*$.
This size is determined by the \textit{inversion number} of the permutation corresponding to $\sigma$, defined by
\[
    I_n(\sigma)\coloneqq\#\{(i,j)\in[n]\times[n]\mid i<j,\ \sigma(i)>\sigma(j)\}.
\]
Specifically,
\[
    \big|\Upper_n^*\sigma\Upper_n^*\big|=\big|\Upper_n^*\big|q^{I_n(\sigma)}.
\]
These weights define the Mallows distribution on $\CS_n$.

\begin{definition}[Mallows Distribution~\cite{mallows1957non}]
    The Mallows distribution $\CM_{n,q}$ on $\CS_n$ assigns to each $\sigma\in\CS_n$ the probability
    \[
        \Pr_{\tau\leftarrow\CM_{n,q}}[\tau=\sigma]
        =\frac{\big|\Upper_n^*\sigma\Upper_n^*\big|}{\big|\GL(n,\F_q)\big|}
        =\frac{(q-1)^n q^{I_n(\sigma)}}{\prod_{k=1}^{n}(q^k-1)}.
    \]
\end{definition}

For fixed $q$, repeated insertion sampling generates $\sigma\leftarrow\CM_{n,q}$ in expected time $\oracle(n)$~\cite{lu2014effective}.
To obtain a uniform sample from $\GL(n,\F_q)$, it remains to sample uniformly within the selected double coset.
The following lemma shows that independent uniform triangular factors suffice.

\begin{lemma}\label{bruhat-lemma}
    If $U_1,U_2\leftarrow\unif(\Upper_n^*)$ and $\sigma\leftarrow\CM_{n,q}$ are sampled independently, then
    \[
        U_1\sigma U_2\sim\unif(\GL(n,\F_q)).
    \]
\end{lemma}

\begin{proof}
    Write $G=\GL(n,\F_q)$ and $B=\Upper_n^*$.
    Conditioned on any fixed $\sigma$, the pair $(AU_1,U_2C)$ has the same distribution as $(U_1,U_2)$ for every $A,C\in B$.
    Thus the distribution of $U_1\sigma U_2$ is invariant under left and right multiplication by $B$.
    Since these operations act transitively on $B\sigma B$, the product is uniform on that double coset.

    For any $M\in G$, let $s\in\CS_n$ be the unique permutation with $M\in BsB$.
    By the definition of $\CM_{n,q}$ and the conditional uniformity above,
    \[
        \Pr[U_1\sigma U_2=M]
        =\Pr[\sigma=s]\cdot\frac{1}{|BsB|}
        =\frac{|BsB|}{|G|}\cdot\frac{1}{|BsB|}
        =\frac{1}{|G|}.
    \]
    Hence $U_1\sigma U_2$ is uniform over $G$.
\end{proof}

For the trapdoored construction, we generate $U_1$ and $U_2$ independently using \cref{alg:invertible-upper} and sample $\sigma$ independently from $\CM_{n,q}$.
A standard hybrid argument, replacing the two triangular factors one at a time and using \cref{bruhat-lemma}, shows that $U_1\sigma U_2$ is computationally indistinguishable from a uniform element of $\GL(n,\F_q)$.

Concretely, we independently sample $(U_i,C_{\mult,U_i},C_{\inv,U_i})$ using \cref{alg:invertible-upper} for $i\in\{1,2\}$, and sample $\sigma$ from the Mallows distribution $\CM_{n,q}$ independently of both tuples.
Since $\sigma$ and $\sigma^{-1}$ simply permute coordinates, their actions admit circuits $C_{\mult,\sigma}$ and $C_{\inv,\sigma}$ of size $\oracle(n)$, constructible in $\oracle(n)$ time.
Setting $G=U_1\sigma U_2$, we define its trapdoor circuits by
\begin{align*}
    C_{\mult,G} &= C_{\mult,U_1}\circ C_{\mult,\sigma}\circ C_{\mult,U_2},\\
    C_{\inv,G} &= C_{\inv,U_2}\circ C_{\inv,\sigma}\circ C_{\inv,U_1}.
\end{align*}
These circuits compute multiplication by $G$ and $G^{-1}$, respectively.
The resulting distribution preserves the sampling and evaluation bounds of the triangular construction.
Finally, \cref{lem:Tellegen} supplies circuits for $\mult^\tp$ and $\inv^\tp$ with the same asymptotic bounds.
Together with the pseudorandomness argument above, this proves \cref{trapdoor GL}.

Our construction has two further consequences: exact determinant evaluation for sampled matrices and efficient polynomial actions after a pseudorandom change of basis.

\begin{corollary}[Efficient evaluation of determinants]\label{cor:determinant}
    Let $M$ be sampled by the construction in \cref{trapdoor GL}.
    Given its trapdoor, $\det M$ can be computed exactly
    in time $\oracle(n)$.
\end{corollary}

\begin{proof}
    Write $M=U\sigma V$ as in the construction of \cref{trapdoor GL}.
    The sampler already knows the diagonal entries of $U,V$ and the permutation $\sigma$.
    Let $c(\sigma)$ be the number of cycles of $\sigma$, including fixed points.
    Traversing these cycles takes $\oracle(n)$ time and gives $\det\sigma=(-1)^{n-c(\sigma)}$.
    Hence
    \[
        \det M=(-1)^{n-c(\sigma)}
        \left(\prod_{i=1}^n U_{ii}\right)
        \left(\prod_{i=1}^n V_{ii}\right)
    \]
    can be computed exactly in $\oracle(n)$ time.
    Retaining this value with the trapdoor preserves the sampling bounds and the matrix marginal.
\end{proof}

Fast forward and inverse actions also allow efficient evaluation of matrix powers after a pseudorandom change of basis.
We apply this to block-diagonal matrices with small blocks, whose powers can be computed separately.

\begin{corollary}[Efficient powers of pseudorandom conjugates]\label{cor:polynomial-actions}
    Fix a polynomial-time constructible block-diagonal matrix $B=\operatorname{diag}(B_1,\ldots,B_m)\in\Mat_{n\times n}(\F_q)$ whose blocks have size at most $r=\operatorname{polylog}(n)$.
    Under the strong (dual) ring-LPN assumption, sample $G$ with its trapdoor from \cref{trapdoor GL} and set $A=GBG^{-1}$.
    Then $A$ is computationally indistinguishable from a uniform element of the conjugacy class
    \[
        \operatorname{Cl}(B)=\{PBP^{-1}:P\in\GL(n,\F_q)\}.
    \]
    Given the blocks of $B$ and the trapdoor for $G$, we can compute $A^k\vect{v}$ in $\ql$ time for every $\vect{v}\in\F_q^n$ and every integer $1\leq k\leq n^{\oracle(1)}$, including the computation of the block powers.
\end{corollary}

\begin{proof}
    For uniform $P\in\GL(n,\F_q)$, the matrix $PBP^{-1}$ is uniform on $\operatorname{Cl}(B)$.
    Since $B$ is polynomial-time constructible and conjugation is polynomial-time computable, the pseudorandomness of $G$ gives the first claim.
    For the power action, use
    \[
        A^k\vect{v}
        =G\bigl(\operatorname{diag}(B_1^k,\ldots,B_m^k)(G^{-1}\vect{v})\bigr).
    \]
    Writing $r_i$ for the size of $B_i$, repeated squaring computes the block powers in
    $\oracle(\sum_i r_i^3\log(k+1))=\oracle(nr^2\log(k+1))=\ql$ time.
    Applying the precomputed blocks takes $\oracle(\sum_i r_i^2)=\oracle(nr)=\ql$ time, and the two changes of basis also take $\ql$ time using the trapdoor.
    Correctness holds for every sample, and the resource bound follows from \cref{trapdoor GL}.
\end{proof}

Under the standard (dual) ring-LPN assumption, the same evaluation guarantee holds in $\nl$ time for any fixed $0<\delta<1$.
More generally, the identity $f(A)=Gf(B)G^{-1}$ transfers any linear circuit for $f(B)\vect{v}$ to one for $f(A)\vect{v}$ with only the cost of the two changes of basis added.
The pseudorandomness guarantee is relative to $\operatorname{Cl}(B)$: conjugation preserves the characteristic polynomial and other similarity invariants of $B$.

\subsection{Depth-Efficient Implementation}\label{sec:depth-efficient}

We modify \cref{alg:efficient-inversion,alg:invertible-upper}
so that all trapdoor circuits have poly-logarithmic depth, while preserving
the preceding sampling and circuit-size bounds.
The key change is to sample each block in the factored form
\[
    U^*=\begin{bmatrix}U_1^*&U_1^*M\\0&U_2^*\end{bmatrix}
    =\begin{bmatrix}
        U_1^* & \vect{0} \\
        \vect{0}   & U_2^*
    \end{bmatrix}
      \begin{bmatrix}I_{d_1}& M\\\vect{0} &I_{d_2}\end{bmatrix}.
\]
Then, we have 
\[(U^*)^{-1} = \begin{bmatrix}
    I_{d_1} & -M \\
    \vect{0} & I_{d_2}
\end{bmatrix} \begin{bmatrix}
    (U_1^*)^{-1} & \vect{0} \\
    \vect{0} & (U_2^*)^{-1}.
\end{bmatrix}\]
For the inverse matrix--vector multiplication we first applies the two diagonal inverses in parallel and then
subtracts $M$ times the lower output from the upper output.
We measure classical depth using bounded-fan-in linear arithmetic circuits
over the fixed field $\F_q$, with unrestricted fan-out. Formally, we state as following theorem:

\begin{theorem}[Depth Optimization]\label{thm:depth-efficient-matrices}
    \cref{Trapdoored Invertible Triangular,trapdoor GL}
    admit modified samplers with the same sampling and circuit-size bounds,
    whose circuits for $\mult$, $\mult^\tp$, $\inv$, and $\inv^\tp$
    have depth $\oracle(\log^2 n)$.
\end{theorem}

The recursion continues to dimension $1$, but uses uniform off-diagonal
blocks whenever $d\leq b$; the hardness assumption is used only above this
cutoff. This also avoids constructing dense inverse matrices at the leaves. 
We also assume that the initial dimension $n$ is large power of $2$.  

\begin{algorithm}
    \caption{Sampling Depth-Efficient Trapdoored Matrices over $\Upper_n^*$}\label{alg:4}
    \begin{algorithmic}[1]
        \Require Initial dimension $n$(assumed power of 2), with $c$ and $b$ as in \cref{alg:strict-upper}.
        \Procedure{OptimizedInvertibleUpper}{$d$}
        \If{$d=1$}
        \State Sample $a\leftarrow\unif(\F_q\setminus\{0\})$.
        \State \Return $([a],\,x\mapsto ax,\,y\mapsto a^{-1}y)$
        \EndIf
        \If{$d\leq b$}
        \State $M\leftarrow\unif(\Mat_{d/2\times d/2}(\F_q))$
        \State Construct a size $\oracle(d^2)$-circuit $C_{\mult, M}$
        \Else
        \State $(M,C_{\mult,M})\leftarrow\D_\CM(d_1,d_2,(2\log^c n)/d)$
        \EndIf
        \State For $i = 1, 2$, $(U_i^*,C_{\mult,U_i^*},C_{\inv,U_i^*})\leftarrow\Call{OptimizedInvertibleUpper}{d/2}$
        \State $U^* \leftarrow \begin{bmatrix}
        U_1^* & \vect{0} \\
        \vect{0}   & U_2^*
        \end{bmatrix}\begin{bmatrix}I_{d_1}&M\\0&I_{d_2}\end{bmatrix}$
        \State Define $C_{\mult,U^*}$ on input $(\vect{x}_1,\vect{x}_2)$ by:
        \Statex \hspace{4.5em} $\vect{z}\leftarrow\vect{x}_1+C_{\mult,M}(\vect{x}_2)$;
        \Statex \hspace{4.5em} output $\bigl(C_{\mult,U_1^*}(\vect{z}),C_{\mult,U_2^*}(\vect{x}_2)\bigr)$.
        \State Define $C_{\inv,U^*}$ on input $(\vect{y}_1,\vect{y}_2)$ by:
        \Statex \hspace{4.5em} compute $\vect{v}\leftarrow C_{\inv,U_1^*}(\vect{y}_1)$ and $\vect{w}\leftarrow C_{\inv,U_2^*}(\vect{y}_2)$ in parallel;
        \Statex \hspace{4.5em} output $\bigl(\vect{v}-C_{\mult,M}(\vect{w}),\vect{w}\bigr)$.
        \State \Return $(U^*,C_{\mult,U^*},C_{\inv,U^*})$
        \EndProcedure
    \end{algorithmic}
\end{algorithm}

\begin{proof}
    The correctness of each circuit is clear. 
    For the pseudorandomness, since $U_1^*$ is always invertible, the map $M \mapsto U_1^* M$ forms a bijection. 
    This implies that the distribution of $U_1^* M$ is computationally indistinguishable from uniform random matrices. As we mentioned, although the recursion continues to dimension 1, for nodes $d \leq b$ is filled by uniform random matrices. Hence, same with the proof of \cref{Trapdoored Invertible Triangular} the distribution of $U^*$ is computationally indistinguishable from uniform distribution. 

    Now we prove the depth bound. At a node of dimension $d>b$, the stored description is $M=HE$,
    where $H=[I_{d_1}\mid C]$ and $C$ is circulant.
    Balanced sums implement both $E$ and $E^\tp$ in depth $\oracle(\log d)$.
    Fast polynomial multiplication~\cite{cantor1991fast}, with the
    circulant coefficients fixed, gives a linear circuit for $C$ of
    size $\oracle(d\log^2 d)$ and depth $\oracle(\log d)$.
    Since $C^\tp$ is also circulant, the same bounds hold for $H$ and
    $H^\tp$.
    At dimensions $1<d\leq b$, balanced matrix--vector
    multiplication gives the same depth bound for $M$ and $M^\tp$.
    In each of the four factorizations above, the two diagonal-block
    actions run in parallel, and the shear adds only $\oracle(\log d)$
    depth. Thus their maximum depth $D(d)$ satisfies
    \[
        D(1)=\oracle(1),\qquad
        D(d)\leq D(d/2) + \oracle(\log d),
    \]
    which gives $D(n)=\oracle(\log^2 n)$.

    Finally, replace the triangular samplers in \cref{2.5} by
    \cref{alg:4}. 
    This proves the claim for \cref{trapdoor GL} as well.
\end{proof}

\section{Trapdoored Clifford Distributions}\label{sec:trapdoored-cliffords}

We now use the trapdoored matrix distributions from \cref{sec:trapdoored-matrices} to construct \emph{trapdoored Clifford Operators}: efficiently samplable Clifford operators together with trapdoors, such that computationally indistinguishable from uniform and the trapdoors support efficient evaluation of selected actions.
In this case, both classical actions and quantum actions are considered. 
The size of quantum circuits are determined by the number of basic Clifford operators including $X, H, P, CX$ and number of ancilla qubits. 

Let $\Pauli$ and $\Clifford$ denote the Pauli and Clifford operators ignoring global phases, respectively.
Throughout this section, the underlying field is $\F_2$.
Clifford operators act on $\Pauli\simeq\F_2^{2n}$ by conjugation.
We order the Pauli generators as $X_1,\ldots,X_n,Z_1,\ldots,Z_n$.
With this convention, a label $(\vect{a},\vect{b})\in\F_2^n\times\F_2^n$ specifies the Pauli operator represented by $X^{\vect{a}}Z^{\vect{b}}$.

\subsection{Uniform random Clifford Oprators and Decomposition of Clifford Operators}\label{sec:uniform-cliffords}

Our construction builds on the uniform Clifford sampler of Bravyi and Maslov~\cite{bravyi2021hadamard}.
We recall the matrix parameters underlying this sampler, with the aim of replacing its uniform matrix distributions by the trapdoored distributions from \cref{sec:trapdoored-matrices}.

The \emph{Hadamard-free Clifford operators} are those Clifford operators that map each computational basis state to another computational basis state, up to a phase.
We denote them by $\CF_n\subseteq\Clifford$.
Up to global phase, every $F\in\CF_n$ is specified by vectors $\vect{a},\vect{b}\in\F_2^n$, a symmetric matrix $\Gamma\in\Sym_n$, and an invertible matrix $\Delta\in\GL(n,\F_2)$.
Its action on computational basis can be written by 
\[
    F\ket{\vect{x}}
    =i^{Q_\Gamma(\vect{x})}(-1)^{\vect{b}^{\tp}\Delta\vect{x}}
    \ket{\Delta\vect{x}\oplus\vect{a}},
    \qquad \vect{x}\in\F_2^n.
\]
Here, $\Delta\vect{x} \oplus \vect{a}$ and $\vect{b}^{\tp}\Delta\vect{x}$ are computed over $\F_2$, while the quadratic phase $Q_\Gamma(\vect{x})$ is defined by
\[
    Q_\Gamma(\vect{x})
    \coloneqq\vect{x}^{\tp}\Gamma\vect{x}
    =\sum_{j=1}^n\Gamma_{jj}x_j
      +2\sum_{1\leq j<k\leq n}\Gamma_{jk}x_jx_k
    \pmod 4.
\]
We denote this operator by $F(\vect{a},\vect{b},\Gamma,\Delta) = F$.

To decompose Clifford operators, we use the \emph{Borel subgroup} $\overline{\mathcal{B}}_n\subseteq\CF_n$, obtained by restricting $\Delta$ to invertible lower-triangular matrix in $\Lower_n^*$.
Let $H_i$ denote the Hadamard gate on the $i$-th qubit, and identify $\CS_n$ with the group of qubit permutations.
The set of Clifford operators admits the Bruhat decomposition~\cite[Theorem~1]{bravyi2021hadamard}
\[
    \Clifford
    =\bigsqcup_{\vect{h}\in\{0,1\}^n}\ \bigsqcup_{S\in\CS_n}
    \overline{\mathcal{B}}_n
    \left(\prod_{i=1}^n H_i^{h_i}\right)S\,\overline{\mathcal{B}}_n.
\]
A convenient variant of the Bravyi--Maslov sampler chooses $(\vect{h},S)$ from their \emph{quantum Mallows distribution}~\cite[Section~3]{bravyi2021hadamard} and, independently, two uniform\footnote{In the original paper, they state the canonical form of each Clifford operators by restricting the parameters of left Borel operator. Since this decomposition is based on Bruhat decomposition, \cref{bruhat-lemma} can be applied here. Hence, we consider that $F_L, F_R \in \overline{\mathcal{B}}_n$ sampled uniformly. } Borel factors $F_L,F_R\in\overline{\mathcal{B}}_n$.
Then $C=F_L(\prod_{i=1}^n H_i^{h_i})SF_R$ is uniform on $\Clifford$: the quantum Mallows law selects each double coset with probability proportional to its size, and the two independent uniform Borel factors give a uniform element within that coset due to \cref{bruhat-lemma}. 

Sampling $(h,S)$ takes $\ql$ time~\cite[Lemma~6]{bravyi2021hadamard}.
Each uniform Borel factor is obtained by independently sampling $\vect{a},\vect{b}\in\F_2^n$, $\Gamma\in\Sym_n$, and $\Delta\in\Lower_n^*$ uniformly, so sampling and recording the full factorization takes $\oracle(n^2)$ time.
In the following subsections, we replace these uniform symmetric and lower-triangular matrices by the trapdoored distributions from \cref{sec:trapdoored-matrices}, while retaining the quantum Mallows law and uniform Pauli corrections.
We first construct trapdoored Hadamard-free Clifford operators and then obtain trapdoored distributions over all Clifford operators.

\subsection{Trapdoored Hadamard-free Clifford Operators}\label{sec:trapdoored-hadamard-free}

We construct trapdoored Clifford operators over $\CF_n$ and $\overline{\mathcal{B}}_n$, supporting both application on quantum states and tableau action on Pauli labels. 
For any $F=F(\vect{a},\vect{b},\Gamma,\Delta)\in\CF_n$, we define the quantum action $\app$ by
\[
    \app(F,\ket{\psi})\coloneqq F\ket{\psi}
\]
for every $n$-qubit state $\ket{\psi}$, extending the computational-basis action in \cref{sec:uniform-cliffords} by linearity.
Conjugation by $F$ maps $X^{\vect{p}_1}Z^{\vect{p}_2}$ to $X^{\vect{p}'_1}Z^{\vect{p}'_2}$ up to a global phase.
The induced classical action on Pauli labels is
\[
    \tab(F,[\vect{p}_1,\vect{p}_2]^\tp)
    \coloneqq
    [\vect{p}_1' , \vect{p}_2']^\tp
    =
    \begin{bmatrix}
        \Delta & 0 \\
        \Delta^{-\tp}\Gamma & \Delta^{-\tp}
    \end{bmatrix}
    \begin{bmatrix}\vect{p}_1\\\vect{p}_2\end{bmatrix}
    =\begin{bmatrix}
        \Delta\vect{p}_1\\
        \Delta^{-\tp}(\Gamma\vect{p}_1+\vect{p}_2)
    \end{bmatrix},
\]
where $\Delta^{-\tp}=(\Delta^{-1})^\tp$.

We use the trapdoor framework of \cref{trapdoor def}, with a quantum circuit $Q_{\app,F}$\footnote{We use $Q$ to denote quantum circuits, to avoid confusion between classical circuits.} and a classical circuit $C_{\tab,F}$.
We say that the trapdoor quantum circuit supports \emph{$T$-efficient quantum application} if, with overwhelming probability over sampling, $Q_{\app,F}$ uses $\oracle(T)$ of basic Clifford gates $X, H, P, CX$ and $r=\oracle(T)$ ancillary qubits, with gates allowed between any pair of qubits.
The sampled quantum circuit $Q_{\app, F}$ must satisfy
\[
    Q_{\app,F}\bigl(\ket{\psi}\otimes\ket{0}^{\otimes r}\bigr)
    =F\ket{\psi}\otimes\ket{0}^{\otimes r}
    \qquad\text{for every $n$-qubit state $\ket{\psi}$}.
\]
In particular, the ancillary qubits are initialized and returned to zero.
Sampling complexity includes producing a succinct description of $F$ and both trapdoor circuits; efficiency for classical action $\tab$ is defined same as in \cref{def:trapdoor-efficiency}.

\begin{theorem}[Trapdoored Hadamard-free Clifford Operators]\label{thm:trapdoored-hadamard-free}
    Let $\CH$ be either $\CF_n$ or $\overline{\mathcal{B}}_n$.
    The following statements hold:
    \begin{enumerate}[label=(\roman*)]
        \item Under the standard (dual) ring-LPN assumption, for any fixed $0<\delta<1$, there exists an $\nl$-samplable trapdoored distribution over $\CH$ supporting $\nl$-efficient quantum application $\app$ and classical evaluation of $\tab$.
        \item Under the strong (dual) ring-LPN assumption, there exists an $\ql$-samplable trapdoored distribution over $\CH$ supporting $\ql$-efficient quantum application $\app$ and classical evaluation of $\tab$.
    \end{enumerate}
    In each case, the distribution of $F$, represented by $F(\vect{a},\vect{b},\Gamma,\Delta)$, is computationally indistinguishable from $\unif(\CH)$, in the sense of \cref{def:trapdoored-matrix}.
\end{theorem}

\begin{proof}
    We give a common construction for both parts, writing $T(n) \in \{\nl, \ql\}$ for the corresponding sampling and circuit-size bound from \cref{sec:trapdoored-matrices}.
    Independently sample $\vect{a},\vect{b}\leftarrow\unif(\F_2^n)$, a uniform diagonal matrix $D$, and $U$ with its trapdoor from \cref{alg:strict-upper}.
    Set $\Gamma=U+U^\tp+D$, retaining the circuit for $U$.
    Independently sample $\Delta$ with its trapdoor: for $\CH=\overline{\mathcal{B}}_n$, transpose a sample from \cref{Trapdoored Invertible Triangular}; for $\CH=\CF_n$, use \cref{trapdoor GL}.
    These constructions supply circuits for $\Delta$, $\Delta^{-1}$, and $\Delta^{-\tp}$ with the stated bounds.

    Replacing $U$ and $\Delta$ by independent uniform matrices, one at a time, gives computationally indistinguishable hybrids by the corresponding matrix guarantees.
    In the final hybrid, $\Gamma$ is uniform over $\Sym_n$, since $(U,D)\mapsto U+U^\tp+D$ is a bijection.
    Thus all four parameters of $F$ are independent and uniform over their respective domains, yielding a uniform element of $\CH$ by the parameterization in \cref{sec:uniform-cliffords}.

    For $C_{\tab,F}$, compute $\Gamma\vect{p}_1$ as $U\vect{p}_1+U^\tp\vect{p}_1+D\vect{p}_1$, and then compute the two output blocks $\Delta\vect{p}_1$ and $\Delta^{-\tp}(\Gamma\vect{p}_1+\vect{p}_2)$.
    The retained circuits and \cref{lem:Tellegen} give both circuit construction and evaluation in $\oracle(T(n))$ time.

    To construct $Q_{\app,F}$, we first make the linear circuits reversible.
    A size-$s$ XOR circuit for a matrix $M$ yields the operation
    \[
        E_M:\ket{\vect{x}}\ket{\vect{y}}
        \longmapsto\ket{\vect{x}}\ket{\vect{y}\oplus M\vect{x}}
    \]
    using $\oracle(s+n)$ CNOT gates and ancillary qubits: compute the intermediate values into fresh qubits, XOR the outputs into the second register, and reverse the intermediate computation.
    All additional workspace is returned to zero and is omitted from the notation.
    The matrix circuits above are linear over $\F_2$, so this applies with $s=\oracle(T(n))$.

    For the quadratic phase, use
    \[
        Q_\Gamma(\vect{x})
        =\sum_{j=1}^n D_{jj}x_j+2(\vect{x}^\tp U\vect{x}\bmod 2)
        \pmod 4.
    \]
    Apply $E_U$ to compute $U\vect{x}$ into a zero register, apply a CZ gate between each pair of corresponding qubits in the two registers, and then reverse $E_U$.
    This produces the phase $(-1)^{\vect{x}^\tp U\vect{x}}$ and resets the second register to zero.
    Applying $P = \operatorname{diag}(1,i)$ to each qubit $j$ with $D_{jj}=1$ supplies the diagonal phase, giving $i^{Q_\Gamma(\vect{x})}$.

    Next, implement the linear change of basis labels by
    \[
        \ket{\vect{x}}\ket{0^n}
        \xrightarrow{E_\Delta}\ket{\vect{x}}\ket{\Delta\vect{x}}
        \xrightarrow{\mathrm{SWAP}}\ket{\Delta\vect{x}}\ket{\vect{x}}
        \xrightarrow{E_{\Delta^{-1}}}\ket{\Delta\vect{x}}\ket{0^n}.
    \]
    Finally, apply $Z^{\vect{b}}$ followed by $X^{\vect{a}}$.
    The resulting action is exactly $F\ket{\vect{x}} \otimes \ket{0^n}$, including its phase, and all ancillary qubits are returned to zero.
    Correctness on arbitrary input states follows by linearity.

    Each sampled linear circuit is used only a constant number of times, with $\oracle(n)$ additional gates for the phases, register swap, and Pauli correction.
    Hence $Q_{\app,F}$ has $\oracle(T(n))$ basic Clifford gates and ancillary qubits, and its gate list can be constructed in $\oracle(T(n))$ time.
    Sampling the succinct parameters and constructing both trapdoor circuits therefore preserve the $\nl$ and $\ql$ bounds in the two cases.
    Correctness holds for every sample; all resource bounds hold with overwhelming probability, as in the underlying matrix constructions.
\end{proof}

The construction of \cref{thm:trapdoored-hadamard-free} also supports the inverse actions $\app^{-1}$ and $\tab^{-1}$ with the same asymptotic bounds.
The inverse quantum action is $\app^{-1}(F,\ket{\psi})\coloneqq F^{-1}\ket{\psi}$, whose computational-basis action is
\[
    F^{-1}\ket{\vect{y}}
    =i^{-Q_\Gamma(\Delta^{-1}(\vect{y}\oplus\vect{a}))}
    (-1)^{\vect{b}^\tp(\vect{y}\oplus\vect{a})}
    \ket{\Delta^{-1}(\vect{y}\oplus\vect{a})}.
\]
The inverse tableau action is
\[
    \tab^{-1}(F,[\vect{p}_1,\vect{p}_2]^\tp)
    =\tab(F^{-1},[\vect{p}_1,\vect{p}_2]^\tp)
    =\begin{bmatrix}
        \Delta^{-1} & 0 \\
        \Gamma\Delta^{-1} & \Delta^\tp
    \end{bmatrix}
    \begin{bmatrix}\vect{p}_1\\\vect{p}_2\end{bmatrix}
    =\begin{bmatrix}
        \Delta^{-1}\vect{p}_1\\
        \Gamma\Delta^{-1}\vect{p}_1+\Delta^\tp\vect{p}_2
    \end{bmatrix}.
\]
Reversing $Q_{\app,F}$ implements $\app^{-1}$ with only constant overhead, since $P^{-1}=P^3$ and the other basic gates are self-inverse.
For $\tab^{-1}$, the retained circuits for $\Gamma$ and $\Delta^{-1}$, together with the circuit for $\Delta^\tp$ obtained by \cref{lem:Tellegen}, give the same $\oracle(T(n))$ construction and evaluation bounds.

Let $F_i=F(\vect{a}_i,\vect{b}_i,\Gamma_i,\Delta_i)\in\CF_n$ for $i=1,2$.
Inversion and composition satisfy the following parameter identities~\cite[Proposition~1]{bravyi2021hadamard}, up to global phase:
\begin{align*}
    F_1^{-1}
    &=F\bigl(\Delta_1^{-1}\vect{a}_1,\vect{b}_{-},
        \Delta_1^{-\tp}\Gamma_1\Delta_1^{-1},\Delta_1^{-1}\bigr),\\
    F_1\circ F_2
    &=F\bigl(\vect{a}_1\oplus\Delta_1\vect{a}_2,\vect{b}_{12},
        \Gamma_2+\Delta_2^\tp\Gamma_1\Delta_2,\Delta_1\Delta_2\bigr).
\end{align*}
Here $F_1\circ F_2$ applies $F_2$ first, and $\vect{b}_{-},\vect{b}_{12}\in\F_2^n$ are the resulting Pauli-$Z$ correction vectors, including the linear phases generated when the quadratic phases are reduced modulo $4$.
All matrix and vector operations in these identities are over $\F_2$.
When $F_1$ and $F_2$ have trapdoors as above with circuit-size bound $T(n)$, the Pauli correction vectors can be recovered in $\oracle(nT(n))$ time by comparing the computational-basis phases at $\vect{0},\vect{e}_1,\ldots,\vect{e}_n$ with the corresponding canonical quadratic phases, where $\vect{e}_j$ is the $j$-th standard basis vector.
The transformed matrix parameters retain $\oracle(T(n))$-size multiplication circuits, since their actions are sums or compositions of a constant number of the original matrix actions; the required inverse and transpose actions for the new $\Delta$ parameters are obtained similarly.
Composing the forward and inverse action circuits already constructed also gives $\oracle(T(n))$-size circuits for $\app$, $\app^{-1}$, $\tab$, and $\tab^{-1}$ of the resulting operators, without requiring explicit recovery of the new Pauli correction vectors.

\subsection{Trapdoored Clifford Operators}\label{sec:full-trapdoored-cliffords}

We now extend the construction to the full family of Clifford operators, obtaining our main result.
Starting from the Bravyi--Maslov sampler recalled in \cref{sec:uniform-cliffords}, we replace its two uniform Borel factors by independent samples from the trapdoored distribution of \cref{thm:trapdoored-hadamard-free}.
The resulting operator distribution remains computationally indistinguishable from uniform, while the retained trapdoors support both quantum application and classical tableau action.
The construction preserves the sampling and evaluation bounds of the Hadamard-free case: sampling the Hadamard and permutation layers takes quasi-linear time, and their actions incur only linear overhead.

\begin{theorem}[Trapdoored Clifford distributions]\label{thm:trapdoored-clifford}
    The following statements hold:
    \begin{enumerate}[label=(\roman*)]
        \item Under the standard (dual) ring-LPN assumption, for any fixed $0<\delta<1$, there exists an $\nl$-samplable trapdoored Clifford distribution over $\Clifford$ supporting $\nl$-efficient quantum application $\app$ and classical evaluation of $\tab$.
        \item Under the strong (dual) ring-LPN assumption, there exists an $\ql$-samplable trapdoored Clifford distribution over $\Clifford$ supporting $\ql$-efficient quantum application $\app$ and classical evaluation of $\tab$.
    \end{enumerate}
    In each case, the operator marginal is computationally indistinguishable from $\unif(\Clifford)$, in the sense of \cref{def:trapdoored-matrix}.
\end{theorem}

\begin{proof}
    We give a common construction for both parts, writing $T(n)\in\{\nl,\ql\}$ for the corresponding sampling and circuit-size bound in \cref{thm:trapdoored-hadamard-free}.
    Independently sample $F_1,F_2\in\overline{\mathcal{B}}_n$ with their trapdoors from that theorem, taking $\CH=\overline{\mathcal{B}}_n$.
    Independently of these factors, sample $(\vect{h},S)$ from the quantum Mallows distribution in $\ql$ time~\cite[Lemma~6]{bravyi2021hadamard}, and set
    \[
        H_{\vect{h}}\coloneqq\prod_{j=1}^n H_j^{h_j},
        \qquad C\coloneqq F_1H_{\vect{h}}SF_2.
    \]

    Replacing $F_1$ and then $F_2$ by independent uniform Borel operators gives computationally indistinguishable hybrids by \cref{thm:trapdoored-hadamard-free}, since the product defining $C$ is polynomial-time computable from classical descriptions of its factors.
    In the final hybrid, $C$ is uniform on $\Clifford$ by the sampler in \cref{sec:uniform-cliffords}.
    Thus the operator marginal is pseudorandom against QPT distinguishers.

    For quantum application, concatenate the circuits in the same order:
    \[
        Q_{\app,C}
        =Q_{\app,F_1}\circ(H_{\vect{h}}S)\circ Q_{\app,F_2},
    \]
    where each factor acts on the data register and its own ancillary register, and is extended by the identity on the other ancillary registers.
    The Hadamard layer uses at most $n$ gates, and $S$ can be implemented using at most $n-1$ SWAP gates, each consisting of three CNOT gates.
    Hence $Q_{\app,C}$ uses $\oracle(T(n))$ basic Clifford gates and ancillary qubits, implements $C$ on every input state, and returns all ancillary qubits to zero.

    For the classical tableau action, let $D_{\vect{h}}=\operatorname{diag}(h_1,\ldots,h_n)$ and let $\Pi_S$ be the permutation matrix satisfying $\Pi_S\vect{e}_j=\vect{e}_{S(j)}$.
    With our ordering of the Pauli labels, the two middle layers have tableaux
    \[
        T_{H_{\vect{h}}}
        =\begin{bmatrix}
            I_n-D_{\vect{h}} & D_{\vect{h}}\\
            D_{\vect{h}} & I_n-D_{\vect{h}}
        \end{bmatrix},
        \qquad
        T_S=\begin{bmatrix}
            \Pi_S & 0\\
            0 & \Pi_S
        \end{bmatrix}.
    \]
    Therefore $T_C=T_{F_1}T_{H_{\vect{h}}}T_ST_{F_2}$.
    Construct $C_{\tab,C}$ by applying these four actions from right to left, using the sampled circuits for $F_1$ and $F_2$.
    The two middle actions only permute Pauli-label bits and take $\oracle(n)$ time, so both construction and evaluation of $C_{\tab,C}$ take $\oracle(T(n))$ time.

    Sampling the succinct factor descriptions and constructing both trapdoor circuits therefore preserve the $\nl$ and $\ql$ bounds in the two cases.
    Correctness holds for every sample; all resource bounds hold with overwhelming probability, as in \cref{thm:trapdoored-hadamard-free}.
\end{proof}

The same modification of \cref{sec:depth-efficient} also gives depth-efficient Clifford implementations. 
A size-$s$, depth-$d$ XOR circuit can be implemented reversibly using
$\oracle(s+n)$ CNOT gates and clean ancillary qubits, with depth
$\oracle((d+1)\log(s+n))$: replace fan-out by binary CNOT copying trees,
compute the intermediate values, copy the outputs, and uncompute.
The strictly upper-triangular circuits used for quadratic phases likewise
have depth $\oracle(\log^2 n)$, using the same shallow off-diagonal
circuits and balanced sums at the leaves.
Thus, using the modified matrix samplers, the constructions of
\cref{thm:trapdoored-hadamard-free,thm:trapdoored-clifford} support quantum application and its inverse
in depth $\oracle(\log^3 n)$, with the same gate-count and ancilla bounds.
Here we use the all-to-all connectivity model of
\cref{sec:trapdoored-hadamard-free}; any qubit permutation is a product
of two layers of disjoint SWAP gates.

\section{Applications}
\subsection{Clifford authentication}\label{subsec: Clifford auth}
We prove that the use of trapdoored Cliffords in authentication gives \emph{information-theoretical} soundness. We begin by recalling the definitions.
\begin{definition}[One-time quantum authentication~\cite{C:DupNieSal12}]
\label{def:application-qas}
A quantum authentication scheme (QAS) consists of a tuple of QPT algorithms
$({\sf KeyGen},{\sf Enc},{\sf Dec})$ where
${\sf Dec}$ outputs a message or $\perp$. Let $k\gets{\sf KeyGen}$.
\begin{itemize}[nosep]
    \item (Completeness)
    ${\sf Dec}_k\circ{\sf Enc}_k=\mathrm{id}_M$ for every key $k$.
    \item ($\eps$-soundness)
    For every attack channel $\Lambda$ on the ciphertext and side
    information $E$, there are completely positive maps
    $\Phi_{\rm acc},\Phi_{\rm rej}$ on $E$, whose sum is trace-preserving,
    such that
    \[
      \frac12\left\|
      \mathbb E_k[({\sf Dec}_k\otimes\mathrm{id}_E)\circ\Lambda
                    \circ({\sf Enc}_k\otimes\mathrm{id}_E)]
      -\mathcal I\right\|_\diamond\le\eps,
    \]
    where $\mathcal I(\rho_{ME})
      =
      (\mathrm{id}_M\otimes\Phi_{\rm acc})(\rho_{ME})
      +
      \ket\perp\bra\perp
      \otimes\Phi_{\rm rej}(\operatorname{Tr}_M\rho_{ME}).$
\end{itemize}
\end{definition}

\begin{definition}[Trapdoored Clifford authentication]
Let $n=\ell+e$ and $\eps=2^{-e}$, with $e\ge1$.
The trapdoored Clifford authentication scheme is as follows.
\begin{description}[nosep]
    \item[{\sf KeyGen}:]
    Generate a trapdoored Clifford $C$, and independently sample
    a uniform random Pauli $P$ and a uniform qubit permutation $\pi$.
    Write $k$ to denote them and $C_k=P\pi C$ for simplicity.
    \item[{\sf Enc}:]
    Apply $C_k$ to $\ket\psi\otimes\ket0^{\otimes e}$.
    \item[{\sf Dec}:]
    Apply $C_k^\dagger$ and measure the last $e$ qubits.
    Return the first $\ell$ qubits if the outcome is $0^e$,
    and $\perp$ otherwise.
\end{description}
\end{definition}

Under the strong (dual) ring-LPN assumption of
\cref{thm:trapdoored-clifford}, key generation, encoding, and decoding
run in time $n^{1+o(1)}$ with overwhelming probability, using the retained
circuit and its reversal.

\begin{theorem}
\label{thm:application-authentication}
Under the pseudorandomness assumption of \cref{thm:trapdoored-clifford},
trapdoored Clifford authentication is a QAS with perfect completeness and
$(\eps+\negl(n))$-soundness.
\end{theorem}

\begin{proof}[Proof sketch]
We use the fact that averaging over $P$ cancels the cross terms in the Pauli expansion
of any attack \cite[Lemma~5.1]{aharonov2017interactive}.
From this, it suffices to consider the attack of the form $\sum_Q\mathcal Q\otimes\Phi_Q$, where
$\mathcal Q(\sigma)=Q\sigma Q^\dagger$, the maps $\Phi_Q$ are completely
positive, $\sum_Q\Phi_Q$ is trace-preserving, and they are independent of $V:=\pi C$.
Thus it suffices to bound the acceptance probability of every $Q\ne I$.

For uniform random Clifford $C$, the Pauli-partitioning lemma and the
counting argument in the proof of \cite[Theorem 1.4]{aharonov2017interactive} give
\begin{align}\label{eqn: probauth}
  \Pr[\text{$V^\dagger QV$ has only $I,Z$ on the last $e$ qubits}]
  =\frac{4^\ell2^e-1}{4^n-1}<2^{-e}
\end{align}
for every $Q\ne I$.
Averaging over the uniform qubit permutation $\pi$ makes this probability
depend only on the numbers of $X$, $Y$, and $Z$ in $Q$.
There are only polynomially many such choices, and the event above can
be tested efficiently using $C$ and $C^\dagger$.

Now we consider the case of trapdoored Cliffords.
Suppose that some $Q$ had acceptance probability (the left-hand side of \cref{eqn: probauth}) at least
$p_0+1/\poly(n)$ for $p_0=(4^\ell2^e-1)/(4^n-1)$.
Note that the acceptance probability is characterized by the numbers of its
$X$, $Y$, and $Z$ factors because of $\pi$. 
There are polynomially many cases, so we can test them to determine if there exists such $Q$. This gives an efficient distinguisher for two cases, contradicting the pseudorandomness.
Thus every $Q\ne I$ has acceptance probability at most
$p_0+\negl(n)<\eps+\negl(n)$ uniformly.
Completeness is immediate.
\end{proof}

\subsection{Worst-case to Average-case Reductions}
We revisit the worst-case to average-case reductions for linear algebra problems and consider some Clifford problems.\footnote{The results of this section requires nonuniform variants of the strong LPN assumption, which was implicit in \cite{vaikuntanathan2026improving}.}

\paragraph{Iterated matrix multiplication.}
We extend the worst-to-average-case reductions for linear algebra \cite{vaikuntanathan2026improving} to products of many invertible matrices.

\begin{theorem}
\label{thm:sequential-matmul}
Let $\F_q$ be a fixed finite field, let $d=d(n)\ge2$ be
polynomially bounded, and let $\eps>0$.
Assume the strong (dual) ring-LPN assumption.
Let $\CA:\GL(n,\F_q)^d\to\F_q^{n\times n}$ be a
randomized polynomial-time algorithm running in time $T(n)$
such that, for every sufficiently large $n$,
\[
  \Pr_{B_1,\ldots,B_d\gets\unif(\GL(n,\F_q))}
  \bigl[\CA(B_1,\ldots,B_d)=B_d\cdots B_1\bigr]
  \ge\eps.
\]
Then we can construct an algorithm that, for any
$A_1,\ldots,A_d\in\GL(n,\F_q)$, computes $A_d\cdots A_1$
with high probability, in time $\widetilde O(T(n))$.
\end{theorem}

\begin{proof}
Let $A_1,\ldots,A_d$ be the input matrices.
Independently sample trapdoored matrices $R_1,\ldots,R_d$
from \cref{trapdoor GL}, and let $R_0=I$.
Define $B_i=R_iA_iR_{i-1}^{-1}$ for $i\in [d]$, which can be computed in time $\widetilde O(dn^2)$ using trapdoors.
We observe that
\[
  B_d\cdots B_1=R_dA_d\cdots A_1.
\]
Thus, whenever $\CA$ is correct,
$Y = R_d^{-1}\CA(B_1,\ldots,B_d)$ equals the desired product.
 
By pseudorandomness, the candidate $Y$ equals $A_d\cdots A_1$ with probability at least
$\eps-\negl(n)$.
We verify correctness with high probability in
$\widetilde O(dn^2)$ time using Freivalds' algorithm.

For sufficiently large $n$, each trial succeeds with probability
at least $\eps/2$. Repeating the procedure
$O(\frac{\log n}\eps)$ times gives the correct product
with high probability. The total running time is
$\widetilde O\left(\frac{T(n)+dn^2}{\eps}\right)$ as desired.
\end{proof}

\paragraph{Determinant computation.}
We improve the determinant reductions of~\cite{vaikuntanathan2026improving}
using the exact determinants provided by \cref{cor:determinant}.

\begin{theorem}
\label{thm:determinant-reduction}
Let $\F_q$ be a fixed finite field, let $\eps>0$, and let
$s_{q,n}=1-\prod_{j=1}^n(1-q^{-j})$.
Assume the strong (dual) ring-LPN assumption.
Let $\CA:\F_q^{n\times n}\to\F_q$ be a randomized
polynomial-time algorithm running in time $T(n)$ such that,
for every sufficiently large $n$,
\[
  \Pr_{B\gets\unif(\F_q^{n\times n})}
  [\CA(B)=\det B]
  \ge
  \begin{cases}
    s_{2,n}+\eps, & q=2,\\
    (1+s_{q,n})/2+\eps, & q>2.
  \end{cases}
\]
Then we can construct an algorithm that, for any
$A\in\F_q^{n\times n}$, computes $\det A$ with high
probability, in time $\widetilde O(T(n))$.
For $q>2$, it also suffices that $\CA$ succeeds with
probability at least $1/2+\eps$ on uniform $\GL(n,\F_q)$.
\end{theorem}

\begin{proof}[Proof sketch]
We adapt the proofs of~\cite[Lemmas~6.6 and~6.7]{vaikuntanathan2026improving},
using the exact determinant of each sampled mask instead of computing it with the average-case algorithm.

For $q>2$, sample a trapdoored $R\in\GL(n,\F_q)$,
compute $\det R$ exactly in $\oracle(n)$ additional time by \cref{cor:determinant},
and set $Z=\CA(RA)/\det R$.
Since a uniform matrix is singular with probability $s_{q,n}$,
the all-matrix promise gives
\[
  \Pr_{B\gets\unif(\GL(n,\F_q))}[\CA(B)=\det B]
  \ge
  \frac{(1+s_{q,n})/2+\eps-s_{q,n}}{1-s_{q,n}}
  =
  \frac12+\frac{\eps}{1-s_{q,n}}
  \ge\frac12+\eps.
\]
First analyze the case where $R$ is uniform on $\GL(n,\F_q)$.
If $A$ is invertible, the matrix $RA$ is
uniform over $\GL(n,\F_q)$, so
$\Pr[Z=\det A]\ge1/2+\eps$.
If $A$ is singular, $\det R$ is uniform over $\F_q^\times$
even conditioned on $RA$, so each nonzero value of $Z$
has probability at most $1/(q-1)\le1/2$.
Repeat this computation independently $t=O(\eps^{-2}\log n)$
times. If some nonzero value appears more than
$(1/2+\eps/2)t$ times among the results, return that value;
otherwise, return zero.

For $q=2$, sample $r$ from the rank distribution of a
uniform binary matrix in $O(n^2)$ time, and let $R=U\operatorname{diag}(I_r,0)V,$
where $U,V$ are independent trapdoored invertible matrices.
Then $\det R=\mathbf1_{\{r=n\}}$.
For uniform $U,V$, the matrix $R$ is uniform.
Writing $u=1-s_{2,n}$, the disagreement probability
$\Pr[\CA(RA)\ne\det R]$ is at most $u-\eps$ when $A$
is invertible, and at least $u$ when $A$ is singular,
by the conditional determinant argument
of~\cite[Lemma~6.7]{vaikuntanathan2026improving}.
Repeat this procedure independently $t=O(\eps^{-2}\log n)$
times. Return one if $\CA(RA)\ne\det R$ occurs fewer than
$(u-\eps/2)t$ times, and zero otherwise.

All relevant events are polynomial-time testable from the sampled matrices, since their determinants are polynomial-time computable.
Thus pseudorandomness changes the probabilities above by only $\negl(n)$.
In both cases, $O(\eps^{-2}\log n)$ independent trials
suffice by concentration.
Each trial takes $T(n)+\widetilde O(n^2)$ time, giving
total time
$\widetilde O\left (\frac{T(n)+n^2}{\eps^2}\right)=\widetilde O(T(n))$.
\end{proof}

\paragraph{Clifford circuit synthesis.}
We give a worst-to-average-case reduction for Clifford circuit synthesis given the corresponding tableau. We say that an algorithm implement a Clifford $C$ if it maps $\ket{\phi}\ket{0^*}$ to $C\ket{\phi}\ket{0^*}$.

\begin{theorem}
\label{thm:clifford-synthesis-reduction}
Let $\eps>0$.
Assume the strong (dual) ring-LPN assumption.
Let $\CA$ be a randomized polynomial-time algorithm running
in time $T(n)$ that, given a signed Clifford tableau $\tau_D$,
outputs a Clifford circuit for $D$ such that, for every sufficiently large $n$,
\[
  \Pr_{D\gets\unif(\Clifford)}
  \bigl[\text{$\CA(\tau_D)$ implements $D$}\bigr]
  \ge\eps.
\]
Then we can construct an algorithm that, given any signed
Clifford tableau, outputs a correct implementation with clean
ancillas with high probability, in time
$\widetilde O(T(n))$.
\end{theorem}

\begin{proof}
Let $C$ be the input Clifford.
Sample a trapdoored Clifford $R$ from
\cref{thm:trapdoored-clifford}.
Define $D=RC$, whose signed tableau can be computed in time
$\widetilde O(n^2)$ using the trapdoor.
Run $\CA(D)$ to obtain a candidate circuit $Q$; whenever $\CA$ implements the correct Clifford, appending the inverse
implementation of $R$ to $Q$ gives an implementation of $C$.
This adds $\widetilde O(n)$ gates and ancillary qubits.

By pseudorandomness, $Q$ correctly implements $D$ with
probability at least $\eps-\negl(n)$.
We verify correctness with high probability in $\widetilde O(T(n)+n^2)$ time using random Pauli conjugation
and ancillary stabilizer tests.

Each trial succeeds with probability
at least $\eps/2$ for large $n$.
Repeating the procedure
$O(\frac{\log n}{\eps})$ times gives a correct implementation
with high probability in time 
$\widetilde O\left(\frac{T(n)+n^2}{\eps}\right)$.
\end{proof}

\subsection{From average-case parallel Clifford implementations to matrix multiplication}
\label{subsec:parallel-clifford}
We reduce worst-case matrix multiplication to constructing parallel
implementations of a random Clifford.
Note that the reduction from the worst-case matrix multiplication to the worst-case parallel Clifford implementations is rather straightforward using $U_A$ below, although we were not aware of this result in the literature. 
The main implication here is that the same complexity persists in the parallel implementation of average-case Cliffords, even though random Cliffords are unlikely to be these classical-like $U_A.$

\begin{theorem}
\label{thm:parallel-clifford}
Let $\eps>0$.
Assume the strong (dual) ring-LPN assumption.
Let $\CA$ be a randomized polynomial-time algorithm running
in time $T(n)$ that, given the signed tableau $\tau_C$ of a
$2n$-qubit Clifford $C$, outputs a Clifford circuit on $2n^2$
qubits such that, for every sufficiently large $n$,
\[
  \Pr_{C\gets\unif(\mathrm{Cliff}(2n))}
  \bigl[\text{$\CA(\tau_C)$ implements $C^{\otimes n}$}\bigr]
  \ge\eps.
\]
Then we can construct a classical algorithm that, for any
$A,B\in\F_2^{n\times n}$, computes $AB$ with high probability,
in time $\widetilde O(T(n))$.
\end{theorem}

\begin{proof}
Let $A,B$ be the input matrices, and define the Clifford $U_A\ket{\vect{u},\vect{v}}
  =
  \ket{\vect{u}+A\vect{v},\vect{v}}.$

Sample a trapdoored $2n$-qubit Clifford $R$ from
\cref{thm:trapdoored-clifford}.
Define $C=RU_A$, whose tableau can be computed in time
$\widetilde O(n^2)$ using the trapdoor.
Run $\CA(\tau_C)$ to obtain a candidate circuit $\Gamma$.

Write $B=[\vect{b}_1\ \cdots\ \vect{b}_n]$, and define
\[
  P_B=\bigotimes_{j=1}^n
  \bigl(I^{\otimes n}\otimes X(\vect{b}_j)\bigr),
  \qquad
  X(\vect{b})=\bigotimes_{k=1}^n X^{b_k}.
\]
Compute a circuit for
$Q=(R^\dagger)^{\otimes n}\Gamma P_B\Gamma^\dagger R^{\otimes n}$,
which is a product of single-qubit Paulis, using the trapdoor
in total time $\widetilde O(T(n)+n^2)$. 

Whenever $\CA$ is correct, we have
\[
  Q=\bigotimes_{j=1}^n
  \bigl(X(A\vect{b}_j)\otimes X(\vect{b}_j)\bigr).
\]
Given the product Pauli structure, we can recover the entries of $A\vect{b}_j$ for each $j$ from the circuit for $Q$ if it is correct.
In other words, the circuit for $Q$ gives a candidate $Y$ for the
product $AB$. 
Whenever $\CA$ is correct, $Y=AB$.

By pseudorandomness, $Y$ equals $AB$ with probability at least
$\eps-\negl(n)$.
We verify correctness with high probability in
$\widetilde O(n^2)$ time using Freivalds' algorithm.
Each trial succeeds with probability at least $\eps/2$
for large $n$.
Repeating the procedure $O(\eps^{-1}\log n)$ times gives
the correct product with high probability in time $  \widetilde O\bigl((T(n)+n^2)/\eps\bigr)
  =\widetilde O(T(n)).$
\end{proof}

We next extend \cref{thm:parallel-clifford} to approximate implementations by general quantum circuits.
\begin{theorem}
\label{thm:parallel-clifford-diamond}
Let $\eps>0$ and $\eta\in[0,1)$ be constants.
Assume the quantum-secure strong (dual) ring-LPN assumption.
Let $\CA$ be a randomized polynomial-time algorithm running
in time $T(n)$ that, given the signed tableau $\tau_C$ of a
$2n$-qubit Clifford $C$, outputs a (possibly non-Clifford) circuit with channel
$\mathcal G_C$ on $2n^2$ qubits such that, for every
sufficiently large $n$,
\[
  \Pr_{C\gets\unif(\mathrm{Cliff}(2n)),\,\CA}
  \left[
    \frac12
    \left\|\mathcal G_C-\mathcal U_{C^{\otimes n}}\right\|_\diamond
    \le\eta
  \right]\ge\eps,
\]
where $\mathcal U_V(\rho)=V\rho V^\dagger$.
Then we can construct a quantum algorithm that, for any
$A,B\in\F_2^{n\times n}$, computes $AB$ with high probability
in time $\widetilde O(T(n))$.
\end{theorem}

\begin{proof}
Define $U_A\ket{\vect{u},\vect{v}}
=\ket{\vect{u}+A\vect{v},\vect{v}}$
as in the proof of \cref{thm:parallel-clifford}.
Sample a trapdoored $2n$-qubit Clifford $R$ from
\cref{thm:trapdoored-clifford}.
Set $C=RU_A$, compute $\tau_C$ in $\widetilde O(n^2)$ time
using the trapdoor, and run $\CA(\tau_C)$.

Write $B=[\vect{b}_1\ \cdots\ \vect{b}_n]$.
Apply $\mathcal G_C$ followed by $(R^\dagger)^{\otimes n}$
to $\bigotimes_{j=1}^n\ket{0^n,\vect{b}_j}$, and measure
the first $n$ qubits of each block to obtain the columns
of a candidate $Y$.
The ideal decoded state is
$\bigotimes_{j=1}^n\ket{A\vect{b}_j,\vect{b}_j}$.
Whenever the approximation condition holds, the diamond-norm
bound changes any measurement event's probability by at
most $\eta$, so $Y=AB$ with probability at least $1-\eta$.

For uniform $R$, the Clifford $C$ is uniform, giving
success probability at least $\eps(1-\eta)$.
The event $Y=AB$ is quantum polynomial-time testable
from $\tau_R$: compute $AB$ classically and synthesize
$R^\dagger$ from its tableau using
\cite[Theorem~8]{aaronson2004improved}.
Thus pseudorandomness gives success probability at least
$\eps(1-\eta)-\negl(n)$ for the trapdoored procedure.

Verify each candidate with high probability using Freivalds'
algorithm and repeat independently
$O((\eps(1-\eta))^{-1}\log n)$ times.
This gives the correct product with high probability in
time $\widetilde O((T(n)+n^2)/(\eps(1-\eta)))
=\widetilde O(T(n))$.
\end{proof}
\newpage

\bibliographystyle{alpha}
\bibliography{cryptobib/abbrev3,cryptobib/crypto,ref}
\appendix
\crefalias{section}{appendix}
\end{document}